%% file: main_arXiv.tex
\documentclass[11pt]{article}
\input{preamble}

\begin{document}

    \title{
        Deliberation via Matching\thanks{This work is supported by NSF grant IIS-2402823.}
    }
    \author{
        Kamesh Munagala\thanks{Department of Computer Science, Duke University, Durham, NC 27708-0129, USA. Email: {\tt \{munagala,\;\;\allowbreak qilin.ye,\allowbreak ian.zhang\}@duke.edu}.}
        \and Qilin Ye\samethanks[2]
        \and Ian Zhang\samethanks[2]
    }

    \maketitle
    \input{abstract}
    \newpage

    \setcounter{tocdepth}{2}
    \tableofcontents

    \thispagestyle{empty}

    \newpage
    \setcounter{page}{1}

    \input{sections/1_introduction}
    \input{sections/2_prelim}
    \input{sections/3_protocol}

    \input{sections/4_warmup}

    \input{sections/5_1_intro}
    \input{sections/5_2_program}
    \input{sections/5_3_coupling}
    \input{sections/5_4_matching}
    \input{sections/5_5_LP}
    \input{sections/5_6_lower_bounds}

    \input{sections/6_sampling}

    \input{sections/7_conclusion}

    \input{acknowledgement}

    \bibliographystyle{plain}
    \bibliography{refs}

    \appendix
    \addtocontents{toc}{\protect\setcounter{tocdepth}{1}}
    \input{sections/appendix_A_algebra}
    \input{sections/appendix_B_dual_fitting}

\end{document}

%% file: preamble.tex
\usepackage[T1]{fontenc}
\usepackage[utf8]{inputenc}
\usepackage[margin=1in]{geometry}
\usepackage[normalem]{ulem}

\usepackage{amsmath,amsthm,mathtools}
\usepackage{newtxtext,newtxmath}

\usepackage[dvipsnames]{xcolor}
\colorlet{LinkColor}{purple}
\colorlet{CiteColor}{Aquamarine}
\colorlet{URLColor}{magenta}
\definecolor{yqlPurple}{HTML}{8346FF}

\usepackage[
  colorlinks=true,
  linkcolor=LinkColor,
  citecolor=CiteColor,
  urlcolor=URLColor,
  linktoc=page
]{hyperref}

\usepackage{graphicx}
\usepackage{tikz}
\usepackage{pgfplots}
\usepackage{import}
\usepackage{multirow}
\usepackage{subcaption}

\usetikzlibrary{decorations.pathreplacing,arrows.meta,calc}
\pgfplotsset{compat=1.18}

\usepackage{aliascnt}

\theoremstyle{plain}
\newtheorem{theorem}{Theorem}[section]

\newaliascnt{lemma}{theorem}
\newtheorem{lemma}[lemma]{Lemma}
\aliascntresetthe{lemma}

\newaliascnt{proposition}{theorem}

\aliascntresetthe{proposition}

\newaliascnt{corollary}{theorem}
\newtheorem{corollary}[corollary]{Corollary}
\aliascntresetthe{corollary}

\newaliascnt{claim}{theorem}

\aliascntresetthe{claim}

\newaliascnt{assumption}{theorem}

\aliascntresetthe{assumption}

\theoremstyle{definition}
\newaliascnt{definition}{theorem}
\newtheorem{definition}[definition]{Definition}
\aliascntresetthe{definition}

\newaliascnt{example}{theorem}
\newtheorem{example}[example]{Example}
\aliascntresetthe{example}

\theoremstyle{remark}
\newaliascnt{remark}{theorem}
\newtheorem{remark}[remark]{Remark}
\aliascntresetthe{remark}

\DeclareMathOperator*{\argmax}{arg\,max}

\usepackage[nameinlink,noabbrev]{cleveref}

\Crefname{theorem}{Theorem}{Theorems}
\Crefname{lemma}{Lemma}{Lemmas}
\Crefname{proposition}{Proposition}{Propositions}
\Crefname{corollary}{Corollary}{Corollaries}
\Crefname{conjecture}{Conjecture}{Conjectures}
\Crefname{observation}{Observation}{Observations}
\Crefname{definition}{Definition}{Definitions}
\Crefname{assumption}{Assumption}{Assumptions}
\Crefname{notation}{Notation}{Notations}
\Crefname{remark}{Remark}{Remarks}
\Crefname{table}{Table}{Tables}
\Crefname{appendix}{Appendix}{Appendices}

\newcommand{\programCref}[1]{%
  {\crefname{equation}{program}{programs}%
   \Crefname{equation}{Program}{Programs}%
   \Cref{#1}}} 

\makeatletter
\def\@fnsymbol#1{\ensuremath{\ifcase#1\or \dagger\or \ddagger\or
   \mathsection\or \mathparagraph\or \|\or **\or \dagger\dagger
   \or \ddagger\ddagger \else\@ctrerr\fi}}
\makeatother
\newcommand*\samethanks[1][\value{footnote}]{\footnotemark[#1]}

%% file: abstract.tex
\begin{abstract}

We study deliberative social choice, where voters engage in small-group discussions to output collective preferences that are then aggregated by a social choice rule.

We introduce a simple \textit{deliberation via matching} protocol. In this protocol, for each pair of candidates, we form a maximum matching among voters who disagree on that pair, and have each matched pair deliberate. We then aggregate the resulting individual and deliberative preferences using the weighted uncovered set tournament rule.

We show that this protocol has a tight distortion bound of $3$ within the metric distortion framework. In the absence of deliberation, general deterministic social choice rules can achieve this distortion, whereas deterministic tournament rules face a strictly larger lower bound of $3.11$. Our result closes this gap: Pairwise deliberation allows a tournament-based rule to attain distortion $3$. Conceptually, this shows that tournament rules can match the power of general deterministic social choice rules once they are given the minimal added power of pairwise deliberations.

We prove this bound via a novel bilinear relaxation of the non-linear program capturing optimal distortion, whose vertices we can explicitly enumerate, leading to an analytic proof. Loosely speaking, our key technical insight is that the distortion objective, as a function of metric distances to any three alternatives, is both supermodular and convex. This characterization therefore provides a new analytical tool for studying the distortion of deliberative protocols, and may be of independent interest.

Finally, although our analysis is for the full protocol, we show that this mechanism also admits a lightweight sampling-based implementation, yielding a high-probability approximation to the deterministic guarantee with arbitrary accuracy and low per-voter complexity.

\end{abstract}

%% file: sections/1_introduction.tex
\section{Introduction}

Collective decision-making lies at the core of both democratic governance and algorithmic social choice. Classical models aggregate individual voter preferences over a set of alternatives using a social choice rule. In practice, individuals deliberate and exchange arguments, often resulting in the emergence of collective preferences. A large body of research in deliberative democracy, most notably ``deliberative polling'' and ``citizens’ assemblies,'' shows that when individuals are given balanced information and structured opportunities for deliberation, they can indeed perform such preference aggregation within the deliberating group \cite{fishkin1991democracy,Ingham18}. This underscores the normative intuition that collective decisions should emerge from public reasoning rather than isolated votes.

At the same time, empirical work indicates that deliberation is most effective in small groups. Large assemblies or unstructured online forums often suffer from coordination challenges, conformity pressures, and polarization effects, where participants reinforce existing biases \cite{Asch,Deutsch1955,Price}. In contrast, small, balanced groups promote reasoned exchange of opinions \cite{Ingham18,FlaniganPW23}, while maintaining manageability and diversity of perspectives. Beyond these empirical considerations, small-group deliberation is also more practical in large-scale settings: it can be implemented in parallel, either through many simultaneous discussions among small groups of participants, or via automated or AI-assisted mediators \cite{ChrisBail,ChatBot,stanfordd,LLM,ashkinaze2024}. These advantages motivate theoretical models that capture the benefits of structured, small-group deliberation rather than full-group discussion.

Recent theoretical work has begun to formalize this intuition \cite{FainGMS17,GoelGM25}. In these models, voters engage in local deliberations that aggregate their ordinal preferences, and a social choice rule is applied to the resulting rankings to find the winning outcome. Such frameworks allow us to study the fundamental algorithmic question of whether structured, small-group deliberation provably improves the efficiency of collective decisions.

We study this question using the \emph{metric distortion} framework \cite{AnshelevichBEPS18}, a quantitative model for evaluating the efficiency of social choice rules. In this framework, both voters and alternatives are embedded in an unknown metric space that captures their underlying preferences: Voters prefer alternatives that are closer to them in this latent metric. A social choice rule then selects a single winner, based only on the voters’ ordinal rankings over alternatives. The distortion of a rule measures how far the chosen winner can be from the voters in terms of total distance in the worst case, compared to the welfare-optimal alternative that minimizes the total distance to all voters had the latent metric been revealed. Thus, a smaller distortion indicates a decision rule that better preserves social welfare despite having only ordinal information.

\noindent\paragraph{Tournament Rules.} Within this setting, it is known that any deterministic rule must incur a distortion of at least $3$ \cite{AnshelevichBEPS18,MunagalaW19,Gkatzelis0020,Kizilkaya022}. A prominent and widely studied subclass of such rules is the class of \emph{tournament rules}, which base their decision on the outcomes of pairwise contests between alternatives. These rules are classical, dating back to Ramon Llull in the 1300s \cite{voting-book}, and further, are the simplest type of rules with bounded metric distortion. Tournament rules are appealing for several reasons. First, they have low representation complexity: We only need the vote counts for quadratically many pairwise comparisons between candidates, rather than counts for all possible rankings, which could be exponentially large. Second, as discussed in \cite{CharikarRT025}, such pairwise comparisons are also often a practical constraint, for instance, arising in recent applications like RLHF \cite{christiano2023,zhu2024}, where humans are asked to compare two outputs of a generative model. Finally, tournament rules also admit low-information implementations: pairwise contests can be estimated accurately from voter samples, so each voter need only compare a small number of alternatives.

The simplicity of tournament rules comes with a drawback: Any tournament rule (that only uses pairwise ranking information about candidates) has a lower bound of $3.11$ on distortion \cite{CharikarRT025}, which is worse than the deterministic optimum of $3$.\footnote{The best known upper bound on the distortion of tournament rules is much larger, around $3.93$ \cite{CharikarRT025}.} This motivates the following question:

\begin{quote}
    \emph{Can small-group deliberation, where voters pairwise aggregate preferences through deliberation, improve the distortion of tournament rules while preserving their simplicity?}
\end{quote}

The recent work of \cite{GoelGM25} provided the first affirmative answer in the case of {\em three-person deliberation}. In their model, every group of three voters deliberates between every pair of alternatives, and each group collectively chooses between any two alternatives by favoring the one with the smaller \emph{average distance} in the latent metric, that is, the alternative closer to the group’s barycenter. When the resulting tournament graph is aggregated through the well-known Copeland tournament rule \cite{voting-book,landau}, the authors showed that the distortion of such a three-person deliberation protocol is strictly better than $3$, thereby surpassing the lower bounds for both tournament and general social choice rules without deliberation. This result established that structured local deliberation can provably improve social welfare. However, their analysis relied on solving a high-dimensional non-convex program numerically, and their protocol achieves a distortion of $4.414$ for two-person deliberations. This leaves open the important question of whether pairwise deliberation also strictly improves the distortion of tournament and general social choice rules.

\subsection{Our Contribution: Deliberation via Matching}

In this paper, we propose a novel and natural protocol for deliberation based on \emph{pairwise discussions} (groups of size 2). Our protocol, called \emph{deliberation via matching}, proceeds as follows. For every pair of candidates, we form an arbitrary maximum matching among voters who disagree on their relative ranking, and each matched pair deliberates.\footnote{We later show, in \Cref{subsect:matching}, that the distortion guarantee is independent of the choice of maximum matching, so it is without loss to begin with an arbitrary one.} The result of each deliberation is a function of their pairwise preferences, as captured by the sum of the latent distances in the underlying metric. These pairwise preferences and individual rankings are then aggregated using the $\lambda$-\emph{weighted uncovered set} tournament rule \cite{MunagalaW19,Kempe_2020}, where $\lambda \in [0.5,1]$ is a parameter controlling the strength of dominance required in pairwise comparisons. A scalar parameter $w \geqslant 0$ controls the influence of deliberation: each matched pair contributes weight $w$ to its joint outcome, while unmatched individual votes retain unit weight.

This protocol differs from prior work \cite{GoelGM25}, which required all groups of voters of a fixed size to deliberate, in that (i) the protocol is more natural to state and is simpler to analyze, (ii) it is parsimonious in requiring only one deliberation per voter for each pair of candidates, and (iii) it allows precise control over how individual votes and pairwise deliberations between voters with opposing preferences are weighted when constructing the tournament graph.

\noindent\paragraph{Distortion Bound.} Within the metric distortion framework, we prove the following main theorem:

\begin{theorem}[Main Theorem, proved in \Cref{sec:general}]
    \label{thm:main-informal}
    The deliberation-via-matching protocol with pairwise deliberation achieves a metric distortion of $3$ for an appropriate choice of $(\lambda,w)$.
\end{theorem}

This breaks the $3.11$ lower bound for tournament rules without deliberation and matches the deterministic optimum of $3$ for any social choice rule without deliberation. Conceptually, this shows that in the metric distortion framework, {\em tournament rules are just as powerful as general social choice rules, provided the former rules are given the minimal added power of pairwise deliberations}.

\noindent\paragraph{Per-Voter Complexity.} Our protocol also admits a lightweight sampling-based variant. In \Cref{sec:sampling}, we show that for each pair of candidates $(A,B)$, the corresponding preference strength can be estimated to arbitrary accuracy with high probability from a small, random sample of voters. The required sample size is quadratic in $m$, the number of candidates, and independent of $n$, the electorate's population size. In particular, when $n \gg m^2 \log m$, any given voter participates in ranking \textit{at most one} pair of candidates and in \textit{at most one} pairwise deliberation with high probability.\footnote{This yields an approximation to the deterministic distortion bound with high probability, and is therefore different from the expected distortion bounds for randomized social choice rules in \cite{AnshelevichP17,CharikarWRW24}.} This means the {\em cognitive load} on a voter is small, showcasing the advantage of such a tournament rule (even with deliberations) over general social choice rules (even without deliberations) that require voters to output \textit{full} rankings or make multiple comparisons.

This additionally justifies why we assume that voter preferences remain unchanged during the deliberation step in our model: Even if they were to change for the deliberating sample, the preferences in the population remain unchanged assuming the population is large. In effect, we therefore assume that our deliberation protocol and tournament rule implicitly elicit the strength of voter preferences in the underlying metric space without changing the individual preferences for the population.

\noindent\paragraph{Lower Bounds.} We complement this positive result with several lower bounds in \Cref{sec:2-candidates,sec:main_lb}. In \Cref{sec:2-candidates}, we show that {\em any deterministic social choice rule} that uses the outcomes of individual votes and pairwise deliberations has distortion at least $2$, even for $m = 2$ candidates. We further show that for $m = 2$ candidates, this bound is tight, and the deliberation-via-matching protocol yields the optimal distortion of $2$.

We finally show in \Cref{sec:main_lb} that the bound of $3$ is optimal for the deliberation-via-matching protocol, that is, for any $(\lambda, w)$, there is an instance with distortion at least $3$. This shows that \Cref{thm:main-informal} cannot be improved for this protocol\footnote{The lower bound also applies when the maximum matching is chosen in a benign fashion as opposed to adversarially.}, though we leave open the question of improving the bound via a different protocol that uses pairwise deliberations (but possibly not a tournament social choice rule).

\noindent\paragraph{Technical Contribution: Bilinear Forms, Supermodularity, and Convexity.} Beyond the quantitative bounds, our main technical contribution is to develop a novel analytical method for studying deliberations. As pointed out in \cite{GoelGM25}, the key difficulty in analyzing deliberative protocols is that the distortion objective is the solution to a non-linear, non-convex program over the distribution of voter–candidate distances, often with unbounded support. This contrasts with classical social choice, where distortion typically arises as the solution to a linear program \cite{Kempe_2020}. The resulting non-linearity severely limits both the classes of deliberative protocols that can be analyzed and the intuition one can draw from such analyses.

Our main contribution in \Cref{sec:general} is to show that, for the deliberation-via-matching protocol, this non-linear program can be reformulated as a \emph{bilinear optimization problem}, where the two linear components correspond to voter masses and metric distances, respectively, each with its own constraint set. This reformulation, which relies crucially on the ``matching'' structure of the protocol, enables both an analytic proof of the distortion bound and a clear characterization of the structure of worst-case instances.

In more detail, our first key observation in the reduction is that this objective has a {\em supermodular} structure in the metric distances. This structure allows us to show that the worst-case instance has voter preferences in a {\em monotonic} order, where the relative strengths of preferences to three given candidates are monotonic. We next show that our specific way of writing the objective function is also {\em convex} in the metric variables, which allows us to use Jensen's inequality to create a small collection of groups of voters based on how they are matched in deliberations, and what the outcomes of these matchings are. We collapse these groups into singleton weighted voters, yielding a bilinear objective with a small number of probability masses, and separate linear constraints on these masses and the metric distances. We then enumerate all vertices of the former polytope (at most six) and solve the resulting linear programs to show our distortion bound. We explicitly produce the corresponding dual certificates, yielding a fully analytic proof of the distortion bound.

As a warm-up, we analyze the special case with only two alternatives in \Cref{sec:2-candidates}. In this setting, the deliberation-via-matching rule admits a simpler and more direct analysis: by pairing voters who disagree and letting each pair support the alternative with the smaller total distance, we show that any winner must be backed, in effect, by at least two-thirds of the electorate. This immediately yields an optimal distortion bound of $2$, improving upon the classic bound of $3$ for deterministic rules without deliberation in this case.\footnote{This also yields a bound of $4$ for the multiple candidate case via standard arguments \cite{AnshelevichBEPS18}.} The two-candidate analysis captures the essential geometry of deliberation and serves as the basis for our general multi-candidate distortion bound in \Cref{sec:general}.

Taken together, our results suggest that pairwise deliberation can be both \emph{powerful} and \emph{tractable}: Even minimal pairwise interactions suffice to make the well-studied tournament rules match the distortion bounds of general social choice rules. More broadly, our bilinear form characterization provides a new method for analyzing deliberative extensions of social choice mechanisms.

\subsection{Related Work}

Our work lies at the intersection of metric distortion, deliberative social choice, and sampling-based decision mechanisms. We briefly touch on the most relevant lines of research.

\noindent\paragraph{Metric Distortion and Tournament Rules.} The metric distortion framework was introduced by Anshelevich et al.\ \cite{AnshelevichBEPS18}, building on earlier work by Procaccia and Rosenschein \cite{ProcacciaR06}, to study how well deterministic voting rules can approximate the social optimum when only ordinal information is available. They showed that the Copeland rule has distortion at most $5$, and that no deterministic rule can achieve distortion below $3$. Later work tightened the upper bound to $3$ via novel social choice rules such as the \emph{matching uncovered set} \cite{Gkatzelis0020,MunagalaW19} and \emph{plurality veto} \cite{Kizilkaya022}. For randomized voting rules, the work of \cite{AnshelevichP17} showed a lower bound of $2$. This lower bound was subsequently improved to $2.11$ by \cite{CharikarR22}. An upper bound of $3$ follows from random dictatorship \cite{AnshelevichP17}, and this was improved to $2.74$ in \cite{CharikarWRW24}. We refer the reader to \cite{AnshelevichFSV21} for a survey.

Tournament rules, which use only pairwise majority contests between candidates, are a central subclass of deterministic voting rules. On the positive side, Munagala and Wang \cite{MunagalaW19} and Kempe \cite{Kempe_2020} defined \textit{weighted tournament rules}, a generalization of the Copeland rule, with distortion at most $2 + \sqrt{5} \approx 4.236$. Subsequently, Charikar et al.\ \cite{CharikarRT025} improved the upper bound to $3.93$ but proved a lower bound of approximately $3.11$ for every deterministic tournament rule, strengthening the lower bound of $3$ from \cite{GoelKM17}. Thus, in the absence of deliberation, tournament rules cannot match the distortion $3$ achievable by general deterministic mechanisms. These rules are fundamentally limited despite their simplicity and centuries-long history \cite{voting-book}.

Our work revisits this barrier using pairwise deliberations. We show that by allowing pairs of voters to aggregate their preferences, a tournament-based rule can in fact achieve a distortion bound of $3$, thereby escaping the $3.11$ bound for non-deliberative tournaments.

\noindent\paragraph{Deliberative Social Choice: From Sortition to Dyads.} The idea that deliberation can improve collective decisions has a long pedigree in political philosophy and deliberative democracy, for example through deliberative polling and citizens’ assemblies \cite{fishkin1991democracy,Ingham18}. Many deliberative systems in practice use \emph{sortition}, which is the random sampling of participants into discussion bodies, to reduce biases and improve legitimacy.

As mentioned before, our deliberation model can be viewed as an implicit mechanism to elicit strength of voter preferences between pairs of candidates. Several theoretical models of eliciting such strengths within the metric distortion framework have been recently proposed \cite{Strength,Distributed,goyal2025metricdistortionprobabilisticvoting}. Focusing on deliberation, Caragiannis et al.\ \cite{CaragiannisM024} examine models of sortition where a large random sample of voters deliberates to compute a consensus or median point, achieving a logarithmic (in the number of alternatives) bound on the sample size required to attain distortion arbitrarily close to one. However, this assumes a single large deliberative body, which raises issues of coordination and bias in practice.

In contrast, our focus is on \emph{small-group deliberation} (more specifically, pairwise deliberations) rather than sortition. Here, Fain et al.\ \cite{FainGMS17} studied a two-person bargaining model under metric preferences, while Goel, Goyal, and Munagala \cite{GoelGM25} studied deliberation by all groups of up to $k$ voters, with the resulting tournament aggregated via the Copeland rule. Their work showed that $k=3$ can beat the deterministic distortion bound, while for $k=2$ they obtained distortion $4.414$. Our protocol improves this to $3$ while requiring one deliberation per voter for each candidate pair.

We note that the protocol in \cite{GoelGM25} required all pairs of voters to deliberate between a pair of alternatives, while our protocol requires only one deliberation per voter for a pair of alternatives.

Finally, focusing on pairwise deliberation is natural, since most discussions unfold through back-and-forth exchanges between pairs of participants. Communication research models such dyadic interactions as the basic unit of conversational dynamics \cite{atomic-dyad,dyadic}, and many deliberative settings, such as in-person debates and online replies, can be viewed as networks of such exchanges \cite{discourse-acts}. Modeling this atomic form of deliberation allows us to capture the core mechanism by which deliberation allows for preference aggregation.

\subsection*{Roadmap}

We formalize the metric distortion framework and deliberation model in \Cref{sec:prelim}, and define the deliberation-via-matching protocol in \Cref{sec:protocol}. We begin with a warm-up analysis of the 2-candidate case in \Cref{sec:2-candidates}. In the technical core of \Cref{sec:general}, we present the analysis of the general case and prove an (upper) bound of $3$ on the distortion of our protocol, along with a matching lower bound in \Cref{sec:main_lb}. In \Cref{sec:sampling}, we describe the sampling-based implementation and show that its sample complexity, measured in the number of sampled voters per candidate pair, is logarithmic in the number of candidates.

%% file: sections/2_prelim.tex
\section{Preliminaries}
\label{sec:prelim}

We begin by reviewing the metric distortion framework and the class of tournament rules used in our analysis, following \cite{AnshelevichBEPS18,MunagalaW19,GoelGM25}.

\noindent\paragraph{Metric Distortion Framework.} Let $\mathcal{C} = \{c_1, \dots, c_m\}$ denote a finite set of $m$ candidates (alternatives), and let $V$ denote a finite set of $n$ voters. Each voter $v \in V$ has a ranking over the candidates that is \emph{consistent} with an underlying latent metric space $(\mathcal{M}, d)$ that contains both voters and candidates as points. If $v$ ranks candidate $\Psi$ higher than $\Upsilon$, then $d(v,\Psi) \leqslant d(v,\Upsilon)$. The metric $d$ is not known to the social planner, who only observes the ordinal rankings induced by it. For any two candidates $\Psi,\Upsilon \in \mathcal{C}$, let $\Psi\Upsilon$ denote the set of voters who prefer $\Psi$ to $\Upsilon$, with cardinality $\lvert \Psi\Upsilon\rvert$. Should ties exist, i.e., $d(v,\Psi) = d(v,\Upsilon)$, we handle them in any consistent way that counts each tied voter toward exactly one of $\Psi\Upsilon, \Upsilon\Psi$. We let $\sigma$ be the profile of preference orderings over candidates for each voter.

For any candidate $\Psi \in \mathcal{C}$, we define its \emph{social cost} with respect to a metric $d$ to be
\[
    SC(\Psi, d) \;=\; \sum_{v\in V} d(v,\Psi).
\]

When the metric $d$ is clear from context, we simply write $SC(\Psi)$. Let $\Psi^\ast = \arg\min_{\Psi \in \mathcal{C}} SC(\Psi)$ denote the socially optimal (1-median) alternative. Given a social choice rule $\mathcal{S}$ that maps the profile of rankings to a winning candidate $\mathcal{S}(\sigma)$, the \emph{distortion} of $\mathcal{S}$ is defined as
\[
    \mathrm{Distortion}(\mathcal{S})
    \;=\;
    \sup_{\sigma}
    \sup_{d \text{ consistent with }\sigma}
    \frac{SC(\mathcal{S}(\sigma), d)}{SC(\Psi^\ast, d)}.
\]

A smaller distortion indicates that $\mathcal{S}$ achieves better welfare despite only knowing ordinal information.

\noindent\paragraph{Tournament Rules.} A \emph{tournament graph} on the candidates is a complete directed graph, with weights $f(\Psi\Upsilon) \in [0,1]$ for each directed edge $\Psi \rightarrow \Upsilon$, so that for every pair of candidates $(\Psi,\Upsilon)$, we have $f(\Psi\Upsilon) + f(\Upsilon\Psi) = 1$. In the setting without deliberation, $f(\Psi\Upsilon)$ represents the fraction of voters that prefer $\Psi$ over $\Upsilon$; however, the weights we construct later will also reflect the outcome of deliberation. A tournament rule takes such a weighted graph as input and outputs the winning candidate.

\label{prelim:lambda-WUS}
Among many tournament-based social choice rules, we focus on the \emph{$\lambda$-weighted uncovered set (WUS)} rule of \cite{MunagalaW19,Kempe_2020}, which builds on the classic uncovered set rules \cite{landau}. Given a tournament with edge weights $f(\Psi\Upsilon) \in [0,1]$, a candidate $\Psi$ is in the \emph{$\lambda$-weighted uncovered set} $\mathrm{WUS}_\lambda$ if for every other candidate $\Upsilon$, either

\begin{enumerate}
    \item $f(\Psi\Upsilon) \geqslant 1-\lambda$, or
    \item there exists a third candidate $\Delta$ such that $f(\Psi\Delta) \geqslant 1-\lambda$ and $f(\Delta\Upsilon) \geqslant \lambda$.
\end{enumerate}

It is known that for $\lambda \in [1/2,1]$, $\mathrm{WUS}_\lambda$ is nonempty \cite{MunagalaW19}. Furthermore, for $\lambda = (\sqrt{5}-1)/2 \approx 0.618$, the rule selecting any candidate from $\mathrm{WUS}_\lambda$ achieves distortion at most $2+\sqrt{5} \approx 4.236$ \cite{MunagalaW19,Kempe_2020}. The special case where $\lambda = 1/2$ is the standard notion of uncovered set \cite{landau}; the classic Copeland rule due to Llull \cite{voting-book} that chooses any candidate that beats the greatest number of others in simple majority voting between them chooses an outcome that belongs to this set.

\noindent\paragraph{Small-Group Deliberation.} We next recall the pairwise deliberation model with \emph{averaging} introduced in \cite{GoelGM25}. A deliberation involves two voters $u,v$ and a pair of candidates $(\Psi,\Upsilon)$. Under the \emph{averaging model}, the pair collectively supports the alternative with smaller total distance, or equivalently,
\[
    \Psi \text{ wins against } \Upsilon
    \quad\text{if }\quad
    d(u,\Psi) + d(v,\Psi) < d(u,\Upsilon) + d(v,\Upsilon).
\]

If equality holds, we handle the deliberation tie in any consistent way that assigns the pair to exactly one of the two outcomes. Importantly, deliberation in our model does not change individual voter preferences or their locations in the metric space; instead, each carefully chosen deliberating group outputs a collective comparison between two candidates, which will be used in our protocol as described in \Cref{sec:protocol}.

%% file: sections/3_protocol.tex
\section{The Deliberation via Matching Protocol}
\label{sec:protocol}

We now describe our main protocol, \emph{Deliberation via Matching}, which implements pairwise deliberation between voters who disagree on a pair of candidates. The protocol defines a weighted tournament over candidates, parameterized by a deliberation weight $w \geqslant 0$ and the $\lambda$-weighted uncovered set parameter $\lambda \in [1/2,1]$. These parameters will be optimized later.

\noindent\paragraph{The Matching Step.} Fix two distinct candidates $\Psi,\Upsilon \in \mathcal{C}$. Let $\Psi\Upsilon$ denote the set of voters who prefer $\Psi$ to $\Upsilon$, and $\Upsilon\Psi$ denote those who prefer $\Upsilon$ to $\Psi$.

Form an arbitrary maximum cardinality matching $M_{\Psi\Upsilon}$ between voters in $\Psi\Upsilon$ and voters in $\Upsilon\Psi$; that is, select $\lvert M_{\Psi\Upsilon}\rvert = \min\{\lvert \Psi\Upsilon\rvert, \lvert \Upsilon\Psi\rvert\}$ disjoint pairs $(u_i,v_i)$ with $u_i \in \Psi\Upsilon$ and $v_i \in \Upsilon\Psi$ for $i=1,\cdots,\lvert M_{\Psi\Upsilon}\rvert$. Each pair $(u_i,v_i)$ represents a pairwise deliberation between two voters with opposing preferences on $(\Psi,\Upsilon)$. Any remaining voters (those not matched) are said to be \emph{unmatched}. Note that all unmatched voters must have the same preference: either they all prefer $\Psi$ (if $\lvert \Psi\Upsilon\rvert \geqslant \lvert \Upsilon\Psi\rvert$) or all prefer $\Upsilon$ (if $\lvert \Psi\Upsilon\rvert < \lvert \Upsilon\Psi\rvert$).

In the averaging model of deliberation, let $W_{\Psi\Upsilon}$ denote the number of matched pairs that favor $\Psi$, and $W_{\Upsilon\Psi} = \lvert M_{\Psi\Upsilon}\rvert - W_{\Psi\Upsilon}$ the number that favor $\Upsilon$.

\noindent\paragraph{The Aggregation Step.} We define the \emph{weighted pairwise score} of $\Psi$ against $\Upsilon$ as
\[
    \mathrm{score}(\Psi\Upsilon;w) = \frac{\lvert \Psi\Upsilon\rvert + w \cdot W_{\Psi\Upsilon}}{n},
\]
and symmetrically $\mathrm{score}(\Upsilon\Psi;w) = (\lvert \Upsilon\Psi\rvert + w \cdot W_{\Upsilon\Psi}) /n$. We divide by $n$ so that the $\mathrm{score}()$ function is independent of $n$, the number of voters. The total score for the pair $(\Psi,\Upsilon)$ is therefore $\mathrm{score}(\Psi\Upsilon;w) + \mathrm{score}(\Upsilon\Psi;w) = 1 + w \cdot \lvert M_{\Psi\Upsilon}\rvert /n$. We define the normalized score to be

\begin{equation}
    f(\Psi\Upsilon;w) = \frac{\mathrm{score}(\Psi\Upsilon;w)}{\mathrm{score}(\Psi\Upsilon;w) + \mathrm{score}(\Upsilon\Psi;w)}
    \label{eq:f}
\end{equation}

and define $f(\Upsilon\Psi;w)$ likewise so that $f(\Psi\Upsilon;w) + f(\Upsilon\Psi;w) = 1$. When the context is clear (e.g. $w$ is a prescribed constant), we may simply write $f(\Psi\Upsilon; w)$ and $\mathrm{score}(\Psi\Upsilon;w)$ as $f(\Psi\Upsilon)$ and $\mathrm{score}(\Psi\Upsilon)$.

Applying this procedure to every ordered pair of candidates $(\Psi,\Upsilon)$ defines a weighted tournament graph on $\mathcal{C}$ where the weight on edge $(\Psi,\Upsilon)$ is $f(\Psi\Upsilon; w)$. The final collective decision is obtained by applying the $\lambda$-weighted uncovered set rule $\mathrm{WUS}_\lambda$ (as defined in \Cref{sec:prelim}) to this tournament.

\noindent\paragraph{The Parameters.} The protocol is governed by two parameters:

\begin{itemize}
    \item the \emph{deliberation weight} $w \geqslant 0$, controlling the relative influence of pairwise deliberation outcomes versus individual preferences, and
    \item the \emph{uncovering parameter} $\lambda \in [1/2,1]$, which determines the strength of the dominance condition used in the $\lambda$-weighted uncovered set rule.
\end{itemize}

When $w = 0$, the protocol reduces to a standard tournament rule without deliberation. As $w$ increases, the outcomes of matched deliberations receive greater emphasis, interpolating smoothly between non-deliberative aggregation and fully deliberative pairwise refinement.

%% file: sections/4_warmup.tex
\section{Warm-up: The Copeland Rule and Some Lower Bounds}
\label{sec:2-candidates}

We first consider the setting in deliberation-via-matching where we set $\lambda = 0.5$ and $w = 1$. This means the deliberation outcomes are given the same importance as individual votes, and we run the Copeland rule to aggregate the tournament into a winner. In the Copeland rule, candidate $A$ beats $B$ if $f(AB;1) \geqslant 0.5$. The rule outputs any candidate that beats the greatest number of other candidates. We show that this protocol has distortion exactly $4$.

Towards this end, we start by analyzing the setting with only $m=2$ candidates and show a distortion of $2$. Since the Copeland winner lies in the uncovered set \cite{AnshelevichBEPS18}, a standard argument shows that the distortion for any $m \geqslant 2$ candidates will be at most the square of the distortion for two candidates, showing an upper bound of $4$ for general number of candidates. Despite the simplicity of this analysis, we show that the bound of $4$ is tight for this setting of $(\lambda, w)$.

For $m=2$, in the absence of deliberation, it is well known that any deterministic social choice rule has a worst-case distortion of $3$ \cite{AnshelevichBEPS18}. Therefore, our deliberation-via-matching protocol offers a major improvement under this setting. We further show that this bound is tight for $m = 2$ candidates {\em regardless} of the deterministic social choice rule used, or the way pairwise deliberations are constructed. This provides an unconditional lower bound for metric distortion with pairwise deliberations. Similarly, we show a lower bound of $1.5$ for randomized social choice rules.

The setting of $(\lambda,w)$ in this section isolates the geometric effect of pairwise deliberation without the additional complexity of tournament aggregation. It therefore acts as a warm-up for the more general analysis in the following section, where we extend the same reasoning to find the optimal $(\lambda, w)$ that yields distortion $3$.

\subsection*{Warm-up: Preliminaries}

Since this section mainly focuses on the $m=2$ candidate case, we specialize the notation to this setting. Let the candidates be $A$ and $B$, separated by distance $d(A,B)$ in the latent metric. Let $AB$ (respectively $BA$) denote the set of voters who prefer $A$ (respectively $B$), so that $\lvert AB\rvert + \lvert BA\rvert = n$ is the total number of voters. Let $M$ denote the arbitrary matching formed between voters in $AB$ and those in $BA$ according to the deliberation-via-matching protocol. Each matched pair $(u,v) \in M \subseteq AB\times BA$ deliberates between $A$ and $B$ and supports the alternative with the smaller total distance to the pair. Define

\begin{align*}
    M_A &= \{(u,v)\in M: A\text{ wins}\} = \{(u,v)\in M: d(u,A) + d(v,A) < d(u,B) + d(v,B)\}, \\
    M_B &= \{(u,v)\in M: B\text{ wins}\} = \{(u,v)\in M: d(u,A) + d(v,A) > d(u,B) + d(v,B)\}.
\end{align*}

(Recall that we handle ties between $AB,BA$ and between $M_A, M_B$ in any consistent way, meaning each voter prefers exactly one in $\{A,B\}$, and so does each matched pair in $M$.) Observe $M_A, M_B$ partition $M$, and recall that the number of $A$-wins pairs (resp. $B$-win pairs) are $W_A = \lvert M_A\rvert$ (resp. $W_B = \lvert M_B\rvert$) by definition. The electorate now splits into three types of voters: (i) Those that contribute to $A$-wins, grouped as pairs from $AB\times BA$; (ii) those that contribute to $B$-wins, also grouped as pairs; and (iii) unmatched voters, all of whom belong to $AB$ if $\lvert AB\rvert \geqslant \lvert BA\rvert$ and $BA$ otherwise. Ties can be apportioned in any way as long as every tie pair is counted once.

Following the protocol in \Cref{sec:protocol}, we will set $\lambda = 0.5$ and $w = 1$. This means we set
\[
    \mathrm{score}(AB) =\frac{\lvert AB\rvert + W_{A}}{n},
\]
and apply the Copeland rule with $f(AB) = \mathrm{score}(AB) / (\mathrm{score}(AB) + \mathrm{score}(BA))$, so that $A$ is the winner if $\mathrm{score}(AB) \geqslant \mathrm{score}(BA)$, and $B$ is the winner otherwise.

We note that the classic Copeland rule declares $A$ as the winner if and only if $\lvert AB\rvert \geqslant \lvert BA\rvert$; it is well known that this rule, as well as any other deterministic rule relying solely on ordinal information, has distortion $\geqslant 3$ even on two candidates \cite{AnshelevichBEPS18}. With deliberation, we instead declare $A$ as the winner if and only if $\lvert AB\rvert + W_A \geqslant \lvert BA\rvert + W_B$, and we show this simple change leads to an improved distortion of $2$.

\subsection{Analysis of the Copeland Rule for Two Candidates}

Assume, without loss of generality, that $A$ is the winner. To bound the distortion, we aim to maximize $SC(A)/SC(B)$, where $SC(\boldsymbol\cdot)$ denotes the social cost.

\noindent\paragraph{Upper-bounding $SC(A)$.} For every voter $v$, we have by triangle inequality

\begin{equation}
    d(v,A) \leqslant d(v,B) + \mathbf{1}[v\in BA]\cdot d(A,B) =
    \begin{cases}
        d(v,B) & \text{ if } v\in AB \\
        d(v,B) + d(B,A) & \text{ if } v\in BA.
    \end{cases}
    \label{eq:sc-a-helper}
\end{equation}

Based on the outcomes of the matching, we split $SC(A)$ into three sums and analyze them separately:
\[
    SC(A) = \sum_{(u,v)\in M_A}^{} [d(u,A) + d(v,A)] + \sum_{(u,v)\in M_B}^{} [d(u,A) + d(v,A)] +  \sum_{v\text{ unmatched}} d(v,A).
\]

\begin{itemize}
    \item For $(u,v)\in M_A$: as $A$ wins the deliberation, we have $d(u,A) + d(v,A) \leqslant d(u,B) + d(v,B)$.
    \item For $(u,v)\in M_B$: assume $u \in AB$ and $v \in BA$, so that the corresponding applications of \Cref{eq:sc-a-helper} give $d(u,A) + d(v,A) \leqslant d(u,B) + d(v,B) + d(A,B)$.
    \item \Cref{eq:sc-a-helper} is also directly applicable on the sum over unmatched voters.
\end{itemize}

Observe that the total additional copies of $d(A,B)$ that appear in $SC(A)$ equals $W_B$ plus the number of unmatched $BA$ voters; this is equivalent to $\lvert BA\rvert - W_A$. Hence,

\begin{equation}
    SC(A) \leqslant  SC(B) + (\lvert BA\rvert - W_A)\cdot d(A,B).
    \label{eq:sc-a}
\end{equation}

\noindent\paragraph{Lower-bounding $SC(B)$.} For any pair $(u,v)\in M_A$, the deliberation constraint and triangle inequality imply

\begin{equation*}
    \begin{cases}
        d(u,B) + d(v,B) \geqslant d(u,A) + d(v,A) \\
        d(u,A) + d(u,B) \geqslant d(A,B) \\
        d(v,A) + d(v,B) \geqslant d(A,B)
    \end{cases}
    \qquad \implies \qquad d(u,B) + d(v,B) \geqslant d(A,B).
\end{equation*}

We now lower bound $SC(B)$ as follows:

\begin{itemize}
    \item Each $(u,v)\in M_A$ contributes $d(A,B)$ to $SC(B)$, and there are $W_A$ such pairs.
    \item The remaining $\lvert AB\rvert - W_A$ voters in $AB$ each contribute at least $d(A,B)/2$ to $SC(B)$, since $d(v,A) \leqslant d(v,B)$ and $d(v,A) + d(v,B) \geqslant d(A,B)$, which imply $d(v,B) \geqslant d(A,B)/2$.
\end{itemize}

Therefore,

\begin{equation}
    SC(B) \geqslant W_A \cdot d(A,B) + (\lvert AB\rvert -W_A)\cdot d(A,B) /2 = (\lvert AB\rvert+W_A) /2 \cdot d(A,B). \label{eq:sc-b}
\end{equation}

Combining \Cref{eq:sc-a} and \Cref{eq:sc-b}, we see that

\begin{align}
    \frac{SC(A)}{SC(B)}&\leqslant \frac{SC(B) + (\lvert BA\rvert - W_A)\cdot d(A,B)}{SC(B)} \leqslant 1 + \frac{(\lvert BA\rvert - W_A)\cdot d(A,B)}{(\lvert AB \rvert + W_A) /2 \cdot d(A,B)} \nonumber \\
    &= 1 + \frac{2(\lvert BA\rvert-W_A)}{\lvert AB\rvert + W_A} = \frac{2n}{\lvert AB \rvert + W_A} - 1 = \frac{2}{\mathrm{score}(AB; w=1)} - 1. \label{eq:distortion-2-candidates}
\end{align}

We now bound the distortion of the protocol.

\begin{theorem}
    The metric distortion of the deliberation via matching protocol with the Copeland Rule for any two-candidate instance is bounded by $2$.
    \label{thm:2-candidates}
\end{theorem}

\begin{proof}
    By \Cref{eq:distortion-2-candidates}, it suffices to show that if $A$ wins, then $\mathrm{score}(AB)  \geqslant 2/3$. To prove this claim, we first assume $\lvert AB\rvert \leqslant \lvert BA \lvert$, so that $\lvert AB\rvert = W_A + W_B$. Since $A$ is the winner,
    \[
        \lvert AB\rvert + W_A \geqslant \lvert BA\rvert + W_B = \lvert BA\rvert + (\lvert AB \rvert - W_A) = n - W_A \qquad \Rightarrow \qquad 2 W_A \geqslant n - \lvert AB\rvert.
    \]

    On the other hand, we also have $W_A\leqslant \lvert AB\rvert$, so $\lvert AB\rvert\geqslant n/3$. Hence

    \begin{equation*}
        n\cdot \mathrm{score}(AB) = \lvert AB\rvert + W_A \geqslant \lvert AB\rvert + \frac{n - \lvert AB\rvert}{2} = \frac{n + \lvert AB\rvert}{2} \geqslant \frac{n + (n/3)}{2} = \frac{2n}{3}.
    \end{equation*}

    If instead $\lvert AB\rvert\geqslant \lvert BA\rvert$ so that $\lvert BA\rvert = W_A + W_B$ and $\lvert AB\rvert \geqslant n/2$, then since $A$ is the winner,
    \[
        \lvert AB\rvert+W_A\geqslant \lvert BA\rvert+W_B = \lvert BA \rvert+(\lvert BA\rvert-W_A) = 2\lvert BA\rvert-W_A \qquad \Rightarrow \qquad 2W_A \geqslant 2\lvert BA\rvert-\lvert AB\rvert.
    \]

    If $\lvert AB\rvert\geqslant 2n/3$ there is nothing to show, so we assume $n/2\leqslant \lvert AB\rvert \leqslant 2n/3$. In this case, the above inequality becomes $2W_A \geqslant 2(n-\lvert AB\rvert) - \lvert AB\rvert = 2n-3\lvert AB\rvert$. Then,

    \begin{equation*}
        n\cdot \mathrm{score}(AB) = \lvert AB\rvert + W_A \geqslant \lvert AB\rvert + \frac{2n-3\lvert AB\rvert}{2} = \frac{2n-\lvert AB\rvert}{2} \geqslant \frac{2n - 2n /3}{2} = \frac{2n}{3}. \qedhere
    \end{equation*}
\end{proof}

By the uncovered set property of the Copeland rule, the distortion for any number of candidates is upper bounded by the square of the distortion on two candidates \cite{AnshelevichBEPS18}. This directly implies the following corollary.

\begin{corollary}
    \label{cor:2-candidates}
    For any number $m$ of candidates, the deliberation-via-matching protocol with $\lambda = 0.5$ and $w=1$ has distortion at most $4$.
\end{corollary}

\subsection{Lower Bounds}
\label{sec:lb}

We first show the following lower bound on the distortion of {\em any} social choice rule that only uses voter preferences and the outcomes of pairwise deliberations.\footnote{In fact, pairwise deliberation outcomes contain individual preferences as a special case: if a voter deliberates with themselves, the joint preference coincides with the voter’s individual preference. Assuming we have complete control over organizing pairwise deliberations, this input model is equivalent to collecting the joint preferences of all unordered pairs $(u,v)$ of voters, one that is strictly richer than pure ordinal rankings by individual voters. See \Cref{sec:conclusion} also.} In particular, this shows that the bound in \Cref{thm:2-candidates} is tight for $m = 2$ candidates, and cannot be improved by either running the deliberations differently or using a different social choice rule.

\begin{theorem}
    \label{thm:lb1}
    Any deterministic social choice rule that uses individual preferences and the outcomes of pairwise deliberations has distortion at least $2$, even with $m = 2$ candidates.
\end{theorem}

\begin{proof}
    We construct two instances $X$ and $Y$ with two candidates $A$ and $B$, which have the same voter preferences, but $SC(B)/SC(A) = 2$ in $X$ and $SC(A)/SC(B) = 2$ in $Y$. In both instances, the metric is on a line where $A$ is at $-1$ and $B$ is at $1$. For $X$, we place two voters at $A=-1$ and one voter at $B=1$, and set the deliberation between a voter at $-1$ and a voter at $1$ to prefer $B$. For $Y$, we place two voters at $0$ (which prefer $A$) and one voter at $B$.

    The preference profile of the voters and the deliberation profiles are identical for these two instances. Thus no deterministic social choice rule can give distortion better than $2$, regardless of the protocol used for constructing deliberating pairs.
\end{proof}

The same pair of instances shows the following corollary. The proof follows by observing that the best any social choice rule can do on the above instance is randomize equally between $A$ and $B$.

\begin{corollary}
    Any randomized social choice rule that uses individual preferences and the outcomes of pairwise deliberations has distortion at least $1.5$, even with $m = 2$ candidates.
\end{corollary}

We finally show that for the setting of $\lambda = 0.5, w = 1$, the deliberation-via-matching rule has distortion exactly $4$ for any $m \geqslant 3$ candidates, showing the na\"ive analysis in \Cref{cor:2-candidates} is in fact tight.

\begin{theorem}
    \label{thm:raw-WUS}
    The deliberation-via-matching protocol with $\lambda = 0.5, w = 1$ has distortion at least $4$.
\end{theorem}

\begin{proof}
    We construct an instance with $3$ candidates $A$, $B$, and $C$, and $3$ voters. The metric is a line and we place $A$ at $0$, $B$ at $1$, and $C$ at $2$. We place two voters at $B=1$ who prefer $A$ over $C$ and one voter at $C=2$. We also set the deliberations between a voter at $B=1$ and a voter at $C=2$ to prefer $C$. When $w=1$, we see that $f(AC) = f(CB) = 0.5$. Thus, candidate $A$ is in the $\lambda$-uncovered set for $\lambda = 0.5$. Since $SC(A)/SC(B) = 4$, the distortion is at least $4$.
\end{proof}

%% file: sections/5_1_intro.tex
\section{Optimal Distortion Bound: Proof of
  \texorpdfstring{\Cref{thm:main-informal}}
  {\Cref{thm:main-informal}}
  }
\label{sec:general}

We now prove \Cref{thm:main-informal}: For an appropriate choice of $(\lambda, w)$, the deliberation-via-matching protocol achieves distortion $3$. Then in \Cref{sec:main_lb}, we show that this bound is tight for the protocol, that every choice of $(\lambda, w)$ incurs distortion at least $3$.

The proof proceeds through a sequence of reductions. In \Cref{subsect:wus} we argue it suffices to analyze instances with three candidates, $A,C$, and $B$. In \Cref{subsect:program}, we discuss the difficulty of directly working with the voter-candidate distances $d(v,A), d(v,C)$, and $d(v,B)$, as they are induced by a latent metric hidden from our mechanism, and we have no good control over their intertwined dependencies (due to triangle inequalities). To circumvent this, we introduce a reparameterization that extracts only the information we \textit{truly} need for the protocol and separates it from the rest of the metric data. This removes the troublesome metric coupling present in the raw voter-candidate distances. Under this perspective, the objective becomes one of optimizing a bilinear program outlined in \programCref{eq:program}. The technical cores are \Cref{subsect:coupling,subsect:matching,subsect:LP}. Here, we exploit the reparametrization to simplify the structure of worst-case instances, and we finally reduce the remaining optimization problem to two small, finite programs (\programCref{eq:final_lp,eq:final_lp2}) that yield the desired distortion bound of $3$. We give more detailed explanations in the corresponding sections.

\subsection{The $\lambda$-Weighted Uncovered Set}
\label{subsect:wus}

In this section, we first assume a fixed $(\lambda, w)$ and use them implicitly to ease notation. We later optimally choose these parameters in \Cref{subsect:LP} and analyze the distortion with the chosen parameters. We first briefly recall how the tournament graph is defined. For any ordered pair of candidates $(\Psi,\Upsilon)$, we let

\begin{equation}
    \label{eq:fw}
    \mathrm{score}(\Psi\Upsilon;w) = \frac{\lvert \Psi\Upsilon\rvert + w \cdot W_{\Psi\Upsilon}}{n}, \qquad f(\Psi\Upsilon; w) = \frac{\mathrm{score}(\Psi\Upsilon;w)}{\mathrm{score}(\Psi\Upsilon;w) + \mathrm{score}(\Upsilon\Psi;w)}
\end{equation}

as in \Cref{eq:f}, where $\lvert \Psi\Upsilon\rvert$ is the number of voters preferring $\Psi$ to $\Upsilon$, $W_{\Psi\Upsilon}$ is the number of deliberation pairs that favor $\Psi$, and $w\geqslant 0$ controls the weight placed on the deliberative outcomes. We then select a winner using the $\lambda$-weighted uncovered set rule on this tournament by selecting any candidate in the \textit{$\lambda$-weighted uncovered set} $\mathrm{WUS}_\lambda$ as the winner. Throughout this section, we write $f(\Psi\Upsilon)$ and $\mathrm{score}(\Psi\Upsilon)$, with the $w$-dependence implicit whenever the context is clear.

Using the analysis technique for uncovered set tournament rules in \cite{AnshelevichBEPS18,MunagalaW19}, suppose $B$ is the optimal candidate and $A$ is the outcome of our protocol. Then, by definition, either $f(AB) \geqslant 1-\lambda$ directly, or there exists another candidate $C$ such that $f(AC) \geqslant 1 - \lambda$ and $f(CB) \geqslant \lambda$. It therefore suffices to consider just three candidates $A,B,C$, and the worst-case distortion over these can be expressed via the following:

\begin{equation}
    \begin{aligned}
        \text{Distortion \;\;=} \qquad & \sup\;\; \frac{SC(A)}{SC(B)} \\[0.5em]
        \text{Subject to\phantom{pvp}} \qquad
        & \text{ either } \;(f(AB) \geqslant 1 - \lambda), \\
        & \text{ or } \phantom{pvp} \big( f(AC) \geqslant 1 - \lambda \text{ and } f(CB) \geqslant \lambda \big).
    \end{aligned}
    \label{eq:objective}
\end{equation}

Since the first case $f(AB) \geqslant  1 - \lambda$ is a restriction of the second case with $C = B$, it further suffices to upper bound the distortion in the second case only, with $f(AC) \geqslant 1-\lambda$ and $f(CB) \geqslant \lambda$.\footnotemark

\footnotetext{Indeed, we can transform the first case into the second by introducing a candidate $C$ who clones $B$. More precisely, co-locate $B,C$ in the metric space, so $SC(C) = SC(B)$. Break all ordinal ties between $C$ and $B$ in favor of $C$, and break all ordinal and deliberation ties between $A$ and $C$ (if any) the same way as between $A$ and $B$. It then follows that $f(AC) = f(AB) \geqslant 1-\lambda$, and $f(CB) = 1 \geqslant \lambda$, since everyone belongs to the set $CB$.}

As two maximum matchings ($(A,C)$ on $AC\times CA$ and $(C,B)$ on $CB\times BC$) will occur, we introduce the following notation to avoid ambiguity by highlighting both alternatives involved. We say a voter pair $(u,v)$ is an \textbf{$\boldsymbol{A\succ C}$ pair} if the pair deliberates in the $(A,C)$ matching and favors $A$. Similarly, a $C\succ A$ pair is one that deliberates in the $(A, C)$ matching and favors $C$. The $C\succ B$ and $B\succ C$ pairs are defined analogously.

%% file: sections/5_2_program.tex
\subsection{A Mathematical Program for Distortion}
\label{subsect:program}

In this section, we develop a reparameterization that underpins the sequence of reductions in \Cref{subsect:coupling,subsect:matching,subsect:LP}. This reparameterization has two purposes. First, it avoids working directly with the raw voter-candidate distances, whose feasibility restrictions are difficult to track across reductions. Second, it separates the information used by the protocol from the remaining absolute-distance information needed by the objective and the metric feasibility constraints.

To see the difficulty dealing with raw distances, suppose three candidates $A,B,C$ are fixed in the an (unkown) latent metric space and we try to optimize directly over $d(v,A), d(v,B), d(v,C)$ for each voter $v$. These quantities cannot be varied independently, as changing $d(v,A)$ generally changes which values of $d(v,B)$ and $d(v,C)$ can still be realized for the same voter. Carrying these metric feasibility restrictions through every subsequent reduction would greatly obscure the structure of the program.

The way around this is to separate the information used by the protocol from those needed elsewhere (by the objective function and the metric feasibility constraints).\footnotemark The reduction in \Cref{subsect:wus} leaves us with just two tournament constraints on $f(AC), f(CB)$, along with an objective comparing $SC(A)$ against $SC(B)$. Take $A$ versus $C$ for example. An important observation is that the protocol always uses $d(v,A)$ and $d(v,C)$ \textit{simultaneously}. A voter $v$'s individual preference is determined by the sign of $d(v,C) - d(v,A)$, and a pairwise deliberation is determined by the sign of the sum of this quantity over the two matched voters. Thus, the natural coordinates for the tournament constraints are the two signed differences that define the two pairwise scores. The remaining absolute-distance information, required by computing $SC(A)$ and $SC(B)$, as well as the metric constraints, is recorded separately in a third variable. In \Cref{prop:Zmin} we show that $d(\boldsymbol\cdot, C)$ carries exactly the remaining information we need, and this value can be optimized explicitly.

To sum up, our bilinear reparameterization is useful because under this perspective, the protocol constraints are expressed through two \textit{separately defined} signed differences, while everything else is handled through the optimized residual variable, and their relation can be computed (cleanly) in closed-form. We no longer need to keep track of the latent metric explicitly. We now formalize this notion.

\footnotetext{This reparameterization is similar in spirit to the \textit{biased metric} framework \cite{CharikarR22,CharikarRT025,CharikarWRW24} used in the standard metric distortion problem. Both approaches avoid optimizing directly over the raw voter-candidate distances and instead expose the metric information relevant to the distortion program. However, to our knowledge, the two frameworks are not directly compatible. The biased metric framework is designed specifically for ordinal rankings, and it is unclear how to extend that framework to encode the deliberative constraints we consider.}

\begin{definition}
    \label{prop:marginal}
    Given an instance, define three variables $X,Y,Z$ on the electorate $V$ by
    \begin{equation}
        \label{eq:XYZ}
        X(v) = d(v,C) - d(v,A), \qquad Y(v) = d(v,B) - d(v,C), \qquad Z(v) = d(v,C).
    \end{equation}
    Then $X(v)$ quantifies voter $v$'s relative preference between $A$ and $C$, and $Y(v)$ between $C$ and $B$.
\end{definition}

With these variables, $X$ determines the entire $(A,C)$ comparison: Up to fixed tie-handling, $v\in AC$ if and only if $X(v) \geqslant 0$, and a matched pair $(u,v)$ favors $A$ (over $C$) if and only if $X(u) + X(v) \geqslant 0$. Similarly, $Y$ determines the $(C,B)$ comparison completely. Finally, $Z$ allows us to recover the absolute distances via $d(v,A) = Z(v) - X(v), d(v,B) = Z(v) + Y(v)$, and $d(v,C) = Z(v)$. Thus, we obtain all the variables needed for \programCref{eq:objective}.

\noindent\paragraph{Continuum of Voters and the Objective Function.} In the discussion below, for the worst-case analysis, we allow the electorate to be represented by a probability distribution over a finite metric support, normalized to unit mass. This relaxation cannot reduce distortion, since every finite electorate is a special case. We therefore view the voters as forming a distribution over the metric space. We will write $\rho(v)$ as the density of a voter at $v$ and normalize $\sum_{v \in V}\rho(v)$ into unit mass. The variables $X,Y\colon V\to \mathbb{R}$ are fixed by the instance and determine $f(AC), f(CB)$ through their one-dimensional distributions, denoted $\mathcal{D}_X, \mathcal{D}_Y$.

\noindent\paragraph{Rewriting Distortion via $X,Y$, and $Z$.} From \Cref{prop:marginal}, we have $SC(A) /SC(B) = [\mathbb{E}Z - \mathbb{E}X] /[\mathbb{E}Z + \mathbb{E}Y]$, where the expectation is over the distribution of voters over the underlying metric space. We now transform the objective into a linear form, observing for $R>0$ that if $SC(A) /SC(B) = [\mathbb{E}Z - \mathbb{E}X] /[\mathbb{E}Z + \mathbb{E}Y]> R+1$, then

\begin{equation}
    \label{eq:linear-obj}
    \mathbb{E}X + (R+1)\cdot \mathbb{E}Y + R\cdot \mathbb{E}Z < 0.
\end{equation}

We will choose $R$ appropriately and show that the global minimum of the LHS of \Cref{eq:linear-obj} is at least zero, and this will imply a distortion of at most $R+1$.

Note that this objective is bilinear, since both the values $(X,Y,Z)$ and the voter distribution over these values are variables, and the support of $(X,Y,Z)$ can be unbounded. Our main contribution below is to relax the problem so that this support becomes constant, and the constraints capturing $\lambda$-WUS become linear.

\noindent\paragraph{Remark.} At several points, we will use an exchange argument over pairs of voters; these arguments can be extended to the continuum over voters by shifting the probability mass appropriately, and we omit the simple details. Further, since we assumed the metric space is finite, the optimization problem above will also have finite size, with the variables corresponding to metric distances and voter masses. We will transform this program in several steps below, noting that these steps will preserve the finite nature of the program.

\noindent\paragraph{Simplifying $Z$.} We first show that the worst-case instances will use a specific setting of $Z$ as a function of $(X,Y)$. We subsequently analyze properties of this function. Note from \Cref{eq:linear-obj} that given fixed $X$ and $Y$, we should point-wise minimize $Z$ such that $\{(X(v), Y(v),Z(v))\}_{v\in V}$ is still metric feasible in the sense that \Cref{eq:XYZ} can be realized in some latent metric space. This leads to the following key lemma. In the lemma below, by $\| X\|_{\infty}$, we mean $\max_v \lvert X(v)\rvert$.

\begin{lemma}
    \label{prop:Zmin}
    Fix real-valued functions $X,Y$ on the electorate $V$. For any real-valued function $Z$ on $V$, in order for $(X,Y,Z)$ to be realized by some metric $d$ under \Cref{eq:XYZ}, it is necessary and sufficient that
    \begin{equation}
        Z(v) \geqslant Z_{\min}(v) = \max \Bigg\{ \frac{\|X\|_{\infty} + X(v)}{2}, \frac{\|Y\|_\infty - Y(v)}{2}, \frac{\|X+Y\|_{\infty} +X(v) - Y(v)}{2} \Bigg\} \qquad \text{ for all }v.
        \label{eq:Zmin}
    \end{equation}
\end{lemma}

\begin{proof}
    We first prove necessity. Because $d$ is nonnegative, we must have $d(v,C) = Z(v) \geqslant 0, d(v,A) = Z(v)-X(v) \geqslant 0$, and $d(v,B) = Z(v) + Y(v) \geqslant 0$ from \Cref{eq:XYZ}. Triangle inequalities for $(v,A,C)$ imply
    \[
        \lvert d(v,A) - d(v,C)\rvert = \lvert X(v)\rvert \leqslant  d(A,C) \leqslant d(v,A) + d(v,C) = 2 Z(v) - X(v).
    \]
    Taking supremum over the first $\leqslant $ gives $d(A,C) \geqslant \|X\|_{\infty}$; combining with the second $\leqslant $ gives
    \begin{equation}
        \label{eq:Zmin-term1}
        2 Z(v) - X(v) \geqslant \|X\|_{\infty} \qquad \text{ so }\qquad Z(v) \geqslant \frac{\|X\|_{\infty} + X(v)}{2}.
    \end{equation}
    The remaining two terms can be obtained analogously by enforcing triangle inequalities on $(v,B,C)$ and $(v,A,B)$, respectively.

    For sufficiency, assume \Cref{eq:Zmin} and define $d(A,C) = \|X\|_{\infty}, d(B,C) = \|Y\|_\infty$, and $d(A,B) = \|X+Y\|_{\infty}$. Let $d(v,C) = Z(v)$, $d(v,A) = Z(v) - X(v)$, and $d(v,B) = Z(v) + Y(v)$. Then $(A,B,C)$ satisfy triangle inequalities, and for each voter, the inequalities established in the necessity part show triangle inequality: For instance, for $(v,A,C)$, we have
    \[
        \lvert d(v,A) - d(v,C) \rvert = \lvert X(v)\rvert \leqslant \| X \|_{\infty} = d(A,C) \leqslant 2 Z(v) - X(v) = d(v,A) + d(v,C),
    \]
    and likewise for $(v,B,C)$ and $(v,A,B)$, so triangle inequalities also hold among these pairs. Finally, to complete the metric, it remains to specify voter-to-voter distances. Note that the current metric defines a graph on $V \cup \{A, B, C\}$ with edges between every pair of candidates, and between each voter and candidate. Thus for two voters $u \neq v$, we can define $d(u, v)$ to be the distance between $u$ and $v$ in this graph.
\end{proof}

Unless otherwise indicated, given a pair $(X,Y)$ defined on $V$, we will from now on default to defining $Z$ by $Z_{\min}(X,Y)$ as stated in \Cref{eq:Zmin}.

\noindent\paragraph{The Bilinear Objective.} From \Cref{prop:marginal}, the distribution $\mathcal{D}_X$ alone determines the $(A,C)$ matching and thus $f(AC)$; similarly $\mathcal{D}_Y$ determines $f(CB)$. From \Cref{eq:linear-obj}, given $R>0$, the distortion is at most $R+1$ if the following functional is non-negative:
\[
    \Phi_R(X,Y) = \mathbb{E}X + (R+1) \cdot \mathbb{E}Y + R \cdot \mathbb{E}[Z_{\min}(X,Y)].
\]
Combining these observations, we obtain the following mathematical program with bilinear objective:

\begin{equation}
    \begin{array}{rl}
        \text{Minimize} \qquad& \Phi_R(X,Y)=\mathbb{E}X+(R+1)\cdot\mathbb{E}Y+R\cdot\mathbb{E}Z\\
        \text{over} \qquad& X,Y \text{ on } V,\quad Z=Z_{\min}(X,Y)\ \text{ from \Cref{eq:Zmin}}\\[1em]
        \text{Subject to} \qquad& \text{(i) } f(AC)\ \text{is induced by \textit{some} matching determined by }X;\\
        & \text{(ii) } f(CB)\ \text{is induced by \textit{some} matching determined by }Y;\\
        & \text{(iii) } f(AC)\geqslant 1-\lambda,\ \ f(CB)\geqslant \lambda.
    \end{array}
    \label{eq:program}
\end{equation}

Note that the above optimization is both over $(X,Y)$ and the choice of the matchings given $(X,Y)$. We will subsequently show that there exists an optimal matching with a specific form that enables writing the final constraints as a set of linear constraints where the variables capture the distribution of the voters. This will make the entire program bilinear, with separate linear constraints for $(X,Y)$ and for the distribution. In summary, we seek to find the smallest $R^*$ under which the infimum of feasible $\Phi$'s remains non-negative. Define
\[
    \mathrm{OPT}(R) = \inf \{\Phi_R(X,Y): (X,Y) \text{ feasible under \programCref{eq:program}}\}, \qquad R^* = \inf \{R>0: \mathrm{OPT}(R) \geqslant 0\}.
\]
Then the supremum of distortion equals $R^* + 1$.

%% file: sections/5_3_coupling.tex
\subsection{Super-modularity and Counter-monotone Coupling}
\label{subsect:coupling}

In this section, we show a key structural property of $Z_{\min}$ from \Cref{eq:Zmin}: it is \textit{supermodular} as a function of $X,Y$. The objective (\programCref{eq:program}) is minimized when $Z = Z_{\min}$, and by supermodularity, this happens when $X$ is paired counter-monotonically with $Y$, as defined below.

We define a \textbf{coupling} of $X$ and $Y$ to be any joint assignment $\{ (X(v),Y(v)): v\in V\}$ that preserves the distributions $\mathcal{D}_X, \mathcal{D}_Y$ under $\rho$.

By \Cref{prop:marginal}, the $f(\boldsymbol\cdot)$-constraints are oblivious to the choice of coupling. Because $\mathbb{E}X$ and $\mathbb{E}Y$ are coupling-invariant, in the function $\Phi_R(X,Y)$, only the term $Z = Z_{\min}(X,Y)$ may change as we vary the coupling. Below, we prove that whenever two voters $v_1, v_2$ satisfy $X(v_1) < X(v_2)$ and $Y(v_1) < Y(v_2)$, swapping their $Y$-values weakly decreases $\mathbb{E}Z$ and hence the objective. Consequently, the optimal coupling is counter-monotone: descending $X$ values are paired with ascending $Y$ values.

This shows that it suffices to examine instances whose induced variables $X$ and $Y$ from \Cref{prop:marginal} are coupled in this manner, and we will do so once the following lemma is proven.

\begin{lemma}[Counter-monotone Coupling of $X,Y$]
    \label{prop:CM-coupling}
    Fix distributions $\mathcal{D}_X, \mathcal{D}_Y$ on $V$ and $R>0$. Then, over $X\sim \mathcal{D}_X$ and $Y\sim \mathcal{D}_Y$, the objective $\Phi_R(X,Y)$ is minimized by one where $X$ and $Y$ are coupled counter-monotonically: if $X(v_1) \leqslant X(v_2)$ then $Y(v_1) \geqslant Y(v_2)$.
\end{lemma}

Given $X \sim \mathcal{D}_X$ and $Y \sim \mathcal{D}_Y$, their expectations are fixed. Then, as discussed earlier, for any $R$, to minimize $\Phi_R(X,Y)$ in \programCref{eq:program}, it suffices to minimize $\mathbb{E}Z$. We prove this claim via an exchange argument: as long as the coupling involves pairs $(x_1,y_1) = (X(u_1), Y(v_1))$ and $(x_2,y_2) = (X(u_2), Y(v_2))$ with $x_1<x_2$ and $y_1<y_2$, swapping them (pairing $x_1$ with $y_2$ and $x_2$ with $y_1$) does not increase $\mathbb{E}Z$.

This exchange argument, and consequently the entirety of \Cref{prop:CM-coupling}, follows directly from showing submodularity of the associated functions, which we establish now. We define the relevant notions first.

\begin{definition}[Submodular and Supermodular Functions]
    \label[definition]{def:submodular}
    A function $f: \mathbb{R}^2\to \mathbb{R}$ is \textbf{submodular} if for all $x_1 \leqslant x_2, y_1\leqslant y_2$,
    \[
        f(x_1,y_1) + f(x_2,y_2) \leqslant f(x_1, y_2) + f(x_2, y_1).
    \]
    Equivalently, $f$ has decreasing differences in $(x,y)$: for every $x_1 < x_2$, the increment $\Delta_x f(y) = f(x_2,y) - f(x_1, y)$ is nonincreasing in $y$. Analogously, $f$ is \textbf{supermodular} if the inequality holds with $\geqslant$.
\end{definition}

\begin{lemma}
    \label{lem:submodular}
    Fix $A,B,C\in \mathbb{R}$. The function $H(x,y) = \max \{A+x, B+y, C+x+y\}$ is submodular.
\end{lemma}

\begin{proof}
    The graph of $H$ is the upper envelope of three planes, $z = A+x, z = B+y$, and $z = C+y+x$. Partition the $(x,y)$-plane into regions where one of these planes is on top. Then, the boundaries are defined by a horizontal line $A+x = C+x+y$, a vertical line $B+y = C+x+y$, and a diagonal line $A+x = B+y$. Observe that along any vertical line with fixed $x$, $H(x, \boldsymbol\cdot )$ as a function of $y$ has the following shape:

    \begin{itemize}
        \item Below the boundary $y_0(x) = \min_{} \{A-C, x+ (A-B)\}$, the top plane is $z = A+x$, which has slope $0$ in the $y$-direction; and
        \item Above $y_0(x)$, the top plane is either $z= B+y$ or $z = C+x+y$, both having slope $1$ in $y$.
    \end{itemize}

    Consequently, on each vertical line, $H$ is flat in $y$ up to a threshold $y_0(x)$; from there, it increases in $y$ with slope $1$. Crucially, the threshold $y_0(x) = \min_{} \{A-C, x + (A-B)\}$ is also nondecreasing in $x$.

    Now fix $y_1 \leqslant y_2$. From the geometric observation above, for $x$,
    \[
        \Delta_y H(x) = H(x, y_2) - H(x, y_1) =
        \begin{cases}
            0 & y_2 \leqslant y_0(x) \\
            y_2 - y_1 & y_1 \geqslant y_0(x) \\
            y_2 - y_0(x) & y_1 < y_0(x) < y_2.
        \end{cases}
    \]

    Because $y_0(x)$ is nondecreasing in $x$, the function $x \mapsto \Delta_y H(x)$ is nonincreasing: when we slide the vertical line to the right, the threshold $y_0(x)$ can only move up, shrinking the portion of $[y_1,y_2]$ above it. Now take $x_1 \leqslant x_2$. The preceding monotonicity gives
    \[
        H(x_2,y_2) - H(x_2,y_1) = \Delta_y H(x_2) \leqslant \Delta_y H(x_1) = H(x_1, y_2) - H(x_1, y_1). \qedhere
    \]
\end{proof}

We are now ready to prove \Cref{prop:CM-coupling}.

\begin{proof}[Proof of \Cref{prop:CM-coupling}]
    The proof consists of two steps. First, given a \textit{frozen} baseline $c = \|X+Y\|_{\infty}$, along with $\|X\|_{\infty}$ and $\|Y\|_{\infty}$ which are fixed by $\mathcal{D}_X$ and $\mathcal{D}_Y$, a local counter-monotone swap never increases $\mathbb{E}Z$. Indeed, by applying \Cref{lem:submodular}, the mapping
    \[
        (x,-y)\mapsto h_c(x,-y) = \max_{} \{{\|X\|_{\infty} + x}, {\|Y\|_{\infty} + (-y)}, {c + x + (-y)} \},
    \]
    as a function of $x$ and $-y$ is submodular. Flipping the sign of the second argument, we see that $(x,y) \mapsto h_c(x,-y)$ as a mapping of $x$ and $y$ is \textit{super}modular. Therefore, when $x_1 < x_2$ and $y_1 < y_2$ (and hence $(-y_1) > (-y_2)$), we have
    \begin{equation}
        h_c(x_1,y_1) + h_c(x_2,y_2) \geqslant h_c(x_1,y_2) + h_c (x_2,y_1),
        \label{eq:local-swap}
    \end{equation}
    so a local counter-monotone swap weakly decreases $\mathbb{E} [h_c(X,Y)]$, conditioned on $\|X+Y\|_{\infty}$ being fixed.

    Second, we claim that local counter-monotone swaps do not worsen (increase) $\|X+Y\|_{\infty}$. To see this, suppose $x_1 < x_2$ and $y_1 < y_2$. Then we have
    \[
        x_1 + y_1 \leqslant \min(x_1 + y_2, x_2 + y_1) \leqslant \max(x_1 + y_2, x_2 + y_1) \leqslant x_2 + y_2.
    \]
    The maximum absolute value of these four terms will therefore be attained at one of the extremes, which implies
    \[
        \max\{\lvert x_1 + y_1\rvert, \lvert x_2 + y_2\rvert\} \geqslant \max \{\lvert x_1 + y_2\rvert, \lvert x_2 + y_1\rvert\}.
    \]
    In particular, this means the swap does not increase $\| X+Y\|_{\infty}$. As $h_c(x,-y)$ is nondecreasing in $c$, this proves that local swaps indeed always help. We therefore iteratively perform local swaps until no swap is available, which happens precisely when the resulting coupling is counter-monotone. Since the objective did not increase in this process, the proof is complete.
\end{proof}

\begin{remark}[\texorpdfstring{Mapping $V$ onto $[0,1]$}{Mapping V onto [0,1]}]
    \label{rmk:XY-interval}
    A useful alternate perspective, which we will frequently use later, is to view $X,Y$ as real-valued functions on $[0,1]$.
\end{remark}

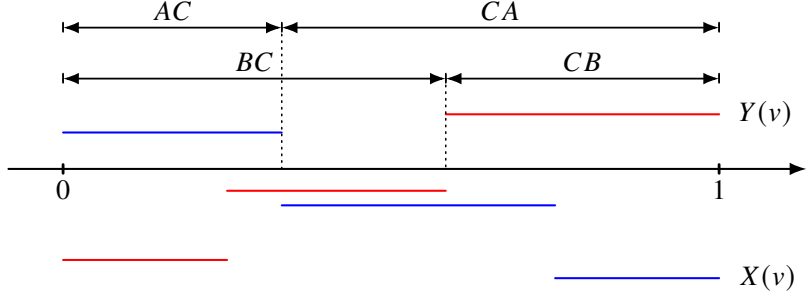
\begin{figure}[htbp]
    \begin{center}
        \resizebox{0.55\linewidth}{!}{\input{figs/counter_monotone}}
    \end{center}
    \caption{Counter-monotonic coupling of $X$ (in blue) and $Y$ (in red). Note the graphs partition $[0,1]$ into $AC/CA$ (by $X$) and $BC/CB$ (by $Y$).}
    \label{fig:monotonic}
\end{figure}

Recall $V$ is normalized into unit mass so it naturally maps to $[0,1]$. Using \Cref{prop:CM-coupling}, we consider a mapping $V\mapsto [0,1]$ such that $X$ (resp. $Y$) can be viewed as a decreasing (resp. increasing) function on $[0,1]$: for instance, we map voters with largest $X(v)$ (and most negative $Y(v)$) to near $0$ and map voters with most negative $X(v)$ (and largest $Y(v)$) to near $1$. \Cref{fig:monotonic} shows one hypothetical example of $X$ and $Y$. As we assumed that the metric space is finite, $X$ and $Y$ will be piecewise constant step functions.

%% file: figs/counter_monotone.tex
\begin{tikzpicture}[>=Latex, line cap=round, xscale=1.5]
    \draw[->, thick] (-0.5,0) -- (6.8,0);
    \draw[line width=1.1pt] (0,-0.06) -- (0,0.06) node[below=2pt] {$0$};
    \draw[line width=1.1pt] (6,-0.06) -- (6,0.06) node[below=2pt] {$1$};
    
    \draw[thick, blue] (0,0.5) -- (2,0.5);
    \draw[dotted, blue] (2,0.5) -- (2,-0.5);
    \draw[thick, blue] (2,-0.5) -- (4.5,-0.5);
    \draw[dotted, blue] (4.5,-0.5) -- (4.5,-1.5);
    \draw[thick, blue] (4.5,-1.5) -- (6,-1.5);
    
    \draw[thick, red] (0,-1.25) -- (1.5,-1.25);
    \draw[dotted, red] (1.5,-1.25) -- (1.5, -0.3);
    \draw[thick, red] (1.5,-0.3) -- (3.5,-0.3);
    \draw[dotted, red] (3.5,-0.3) -- (3.5,0.75);
    \draw[thick, red] (3.5,0.75) -- (6,0.75);
    
    \path (6,-1.5) node[anchor=west, xshift=4pt] {$X(v)$};
    \path (6,0.75) node[anchor=west, xshift=4pt] {$Y(v)$};
    
    \draw[line width=0.6pt] (0,1.89) -- (0,2.01);
    \draw[line width=0.6pt] (2,1.89) -- (2,2.01);
    \draw[line width=0.6pt] (6,1.89) -- (6,2.01);
    \draw[line width=0.6pt] (0,1.88) -- (0,2.00);
    \draw[line width=0.6pt] (2,1.88) -- (2,2.00);
    \draw[line width=0.6pt] (6,1.88) -- (6,2.00);
    \node[above] at (1.0,1.90) {$AC$};
    \node[above] at (4.0,1.90) {$CA$};
    \draw[<->, line width=0.7pt] (0,1.94) -- (2,1.94);
    \draw[<->, line width=0.7pt] (2,1.94) -- (6,1.94);
    
    \draw[line width=0.6pt] (0,1.19) -- (0,1.31);
    \draw[line width=0.6pt] (3.5,1.19) -- (3.5,1.31);
    \draw[line width=0.6pt] (6,1.19) -- (6,1.31);
    \draw[line width=0.6pt] (0,1.18) -- (0,1.30);
    \draw[line width=0.6pt] (3.5,1.18) -- (3.5,1.30);
    \draw[line width=0.6pt] (6,1.18) -- (6,1.30);
    \node[above] at (1.75,1.20) {$BC$};
    \node[above] at (4.75,1.20) {$CB$};
    \draw[<->, line width=0.7pt] (0,1.24) -- (3.5,1.24);
    \draw[<->, line width=0.7pt] (3.5,1.24) -- (6,1.24);
    
    \draw[dotted, line width=0.6pt] (2,1.94) -- (2,0);
    \draw[dotted, line width=0.6pt] (3.5,1.24) -- (3.5,0);
\end{tikzpicture}

%% file: sections/5_4_matching.tex
\subsection{Tight $f$ Constraints and the Optimal Matchings}

\label{subsect:matching}

A fixed instance $I$ may admit various matchings for the candidate pair $(A,C)$ and thus potentially different values for $W_{AC}$, the number of matchings where the outcome is $A \succ C$. Consequently, the values of $f(AC)$ need not be unique; neither for $f(CB)$. However, we note that for fixed $(X,Y)$ and distribution $\rho$ over voters, the constraints for $f(AC)$ and $f(CB)$ in our mathematical program are made most slack by choosing the matchings with the greatest number of $AC$ wins for $X$ (resp. $CB$ wins for $Y$). Call them the \textbf{$\boldsymbol A$-optimal $\boldsymbol{(A,C)}$ matching} and the \textbf{$\boldsymbol C$-optimal $\boldsymbol{(C,B)}$ matching}, respectively.

We now show two properties of the optimal solution. We first prove that a specific type of optimal matchings pairs ``prefixes'' (most polar voters) of one side with the ``suffixes'' (most indifferent, i.e., least polar) of the other. Next, we prove a ``continuity'' result: that it suffices to tighten the inequalities $f(AC) \geqslant 1-\lambda$ and $f(CB) \geqslant \lambda$ into equalities. These results pave the way to a clean structural reduction (\programCref{eq:program-reduced}) that leads us to \Cref{subsect:LP}.

\subsubsection{Property 1: Prefix Property of Matchings}

Fix the instance $I$ as well as some matching for candidates $(A,C)$. Consider two $A\succ C$ pairs $(u_1,v_1)$ and $(u_2, v_2)$. Suppose $X(u_1) \geqslant X(u_2) \geqslant 0 \geqslant X(v_1) \geqslant X(v_2)$. Then we have $X(u_1) + X(v_1) \geqslant 0$ and $X(u_2) + X(v_2) \geqslant 0$. It is easy to check that $X(u_1) + X(v_2) \geqslant 0$ and $X(u_2) + X(v_1) \geqslant 0$. This means we can replace the matchings with $(u_1,v_2)$ and $(u_2,v_1)$. This means the matchings can be made {\em counter-monotone}. Further, suppose $X(u_1) \geqslant X(u_2) \geqslant 0$ and $u_1$ does not participate in an $A\succ C$ pair, while $u_2$ is matched to $v_2$ in an $A\succ C$ pair. Then we can replace $(u_2,v_2)$ with $(u_1,v_2)$.

Analogously, if $0 \geqslant X(v_1) \geqslant X(v_2)$ and $v_1$ does not participate in an $A\succ C$ pair while $v_2$ is matched to $u_2$ in an $A\succ C$ pair, we can replace $(u_2,v_2)$ with $(u_2,v_1)$. This is feasible for the $f$ constraint since $W_{CA}$ cannot increase in this process, and $W_{AC}$ is preserved.

Iterating this process, we obtain a new $(A,C)$ matching that also has $W_{AC}$ pairs satisfying $A\succ C$. Additionally, in this new $(A,C)$ matching, these pairs come from pairing the $W_{AC}$ mass of highest $X(u) \geqslant 0$ (the \textbf{prefix} of $AC$) with the $W_{AC}$ mass with highest $X(v) < 0$ (the \textbf{suffix} of $CA$) counter-monotonically, meaning that between these two blocks of mass $W_{AC}$, the highest positive $X(u)$ is matched to the lowest $X(v) < 0$, and so on. This is shown in \Cref{fig:matching-A-only}.

\begin{figure}[htbp]
    \begin{center}
        \resizebox{0.85\linewidth}{!}{\input{figs/matching_Awin_only}}
    \end{center}
    \caption{The prefix-suffix structure of $A\succ C$ pairs.}
    \label{fig:matching-A-only}
\end{figure}
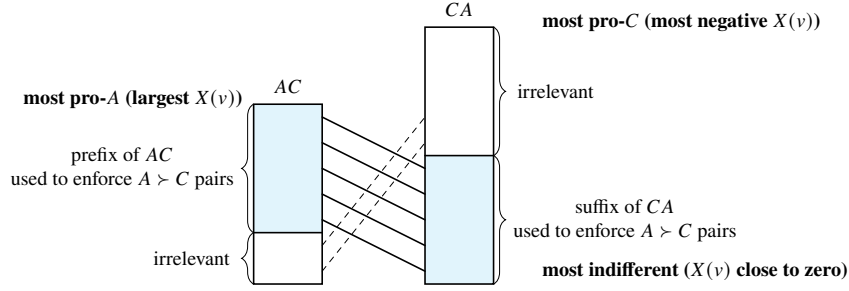

An identical result can be shown for $(C,B)$, so that any $C\succ B$ matchings of mass $W_{CB}$ can be assumed to satisfy the prefix property. In particular, by starting with an optimal matching and repeating the procedure above, we arrive at the following conclusion.

\begin{lemma}[Prefix Property of Optimal Matchings]
    \label{lem:prefix}
    Fix an instance $I$ and an $A$-optimal matching for candidates $(A,C)$. Suppose this matching admits a mass $W_{AC}$ of $A\succ C$ pairs. Then, there exists another $A$-optimal $(A,C)$ matching (hence also with mass $W_{AC}$ of $A\succ C$ pairs) such that:

    \begin{itemize}
        \item[(i)] It takes place between the $W_{AC}$ masses of highest $X(u)\geqslant 0$ and highest $X(v)<0$; and
        \item[(ii)] It couples the two blocks counter-monotonically: highest $X(u)$ with lowest $X(v)$, and so on.
    \end{itemize}

    An equivalent version holds for the $(C,B)$ matching, where if $W_{CB}$ is the largest admissible mass of $C\succ B$ pairs, then they can be assumed to be coupled counter-monotonically between the $W_{CB}$ mass with highest $Y(u)\geqslant 0$ and the $W_{CB}$ mass with highest $Y(v) < 0$.
\end{lemma}

Recall from \Cref{rmk:XY-interval} that we may view $X,Y$ as monotonic functions on $[0,1]$. The previous observations yield a four-interval decomposition: two blocks of size $W_{AC}$, one for each of $AC, CA$, and two complementary blocks. A similar decomposition follows for $(C, B)$.

We focus only on the cases where $\lvert AC\rvert \leqslant \lvert CA\rvert$ and $\lvert BC\rvert \leqslant \lvert CB\rvert$, but the other cases also admit a similar partition, the major difference being the location of the unmatched voters.

\begin{lemma}[\texorpdfstring{Four-interval partition of $[0,1]$ by $X$ and $Y$}{Four-interval partition of [0,1] by X and Y}]
    \label{lem:partition}
    When $\lvert AC\rvert \leqslant \lvert CA\rvert$, the range $[0,1]$ can be partitioned into four consecutive, possibly empty intervals that describe the $(AC, CA)$ matching, as shown in \Cref{tab:interval}. By the structure of the $A$-optimal matching, the $A$-win blocks always lie on the leftmost of the $[0,\lvert AC\rvert]$ ($AC$ block) and $[\lvert AC\rvert, 1]$ ($CA$ block).

    Similarly, when $\lvert BC\rvert \leqslant \lvert CB\rvert$, we can perform the same partition based on $Y$ which describes the $(BC, CB)$ matching.

    Since we analyze the $C$-optimal matching and $Y$ is increasing, the $C$-win blocks lie on the \emph{rightmost} of the $[0, \lvert BC\rvert]$ and $[\lvert BC\rvert, 1]$ blocks.
\end{lemma}

\begin{table}[h]
    \centering
    \renewcommand{\arraystretch}{1.3}
    \setlength{\tabcolsep}{4pt}
    \begin{tabular}{|c||c|c|c|c|c|}
        \hline
        \multirow{3}{*}{$X$} & Interval & $[0,W_{AC}]$ & $[W_{AC},\lvert AC\rvert]$ & $[\lvert AC\rvert,\lvert AC\rvert+W_{AC}]$ & $[\lvert AC\rvert+W_{AC}, 1]$ \\ \cline{2-6}
        & Role   & $AC$ $A$-win & $AC$ $A$-loss & $CA$ $A$-win & $CA$ $A$-loss/unmatched \\ \cline{2-6}
        & Length & $W_{AC}$ & $W_{CA}$ & $W_{AC}$ & $W_{CA} + (1-2\lvert AC\rvert)$ \\
        \hline\hline
        \multirow{3}{*}{$Y$} & Interval & $[0,W_{BC}]$ & $[W_{BC}, \lvert BC\rvert]$ & $[\lvert BC\rvert,1 - \lvert BC\rvert + W_{BC}]$ & $[1 - \lvert BC\rvert + W_{BC},1]$ \\ \cline{2-6}
        & Role   & $BC$ $C$-loss & $BC$ $C$-win & $CB$ unmatched/$C$-loss & $CB$ $C$-win \\ \cline{2-6}
        & Length & $W_{BC}$ & $W_{CB}$ & $(1 - 2\lvert BC\rvert) + W_{BC}$ & $W_{CB}$ \\
        \hline
    \end{tabular}
    \caption{\label{tab:interval}Two different partitions of $[0,1]$ induced by $X$ and $Y$.}
\end{table}

\subsubsection[Property 2: Tightness of f(AC) >= 1-lambda and f(CB) >= lambda]
{Property 2: Tightness of $f(AC) \geqslant 1-\lambda$ and $f(CB) \geqslant \lambda$}

Since the constraints for $f(AC)$ and $f(CB)$ in our mathematical program are made most slack by choosing the $A$-optimal $(A, C)$ matching and $C$-optimal $(C, B)$ matchings, we will now restrict our $f$ values to correspond to the optimal matchings.

Furthermore, we define $M_{AC}$ and $M_{CB}$ to be the optimal $(A, C)$ and $(C, B)$ matchings which satisfy the prefix property in \Cref{lem:prefix}. Similarly $W_{AC}$ and $W_{CB}$ correspond to the number of wins in the optimal matching.

We now show a key result, that without loss of generality, we can assume that both the $f(AC) \geqslant 1 - \lambda$ and $f(CB) \geqslant \lambda$ constraints are tight.

\begin{lemma}
    \label{lem:tight-AC}
    Let $I$ be an instance under which $f(AC) \geqslant 1-\lambda$ and $f(CB) \geqslant \lambda$. Then there exists an instance $\hat{I}$ such that $\hat{f}(AC) = 1-\lambda, \hat{f}(CB) = \lambda$, and $\Phi_R(X', Y') \leqslant \Phi_R(X,Y)$. Consequently, we may assume without loss of generality that $f(AC) = 1-\lambda$ and $f(CB) = \lambda$ in \programCref{eq:program}.
\end{lemma}

\begin{proof}
    The proof consists of two separate, analogous transformations: tightening $f(AC)$ to $1-\lambda$ and tightening $f(CB)$ to $\lambda$.

    We first show that the constraint $f(AC) \geqslant 1-\lambda$ can be made tight while preserving the other constraint and the objective. For $t\geqslant 0$, we define a parameterized instance $I_t$ with variables $X_t(v) = X(v) - t$ and keep $Y_t = Y$ unchanged for all $t$. Since $f(CB)$ is determined by $Y$, $f(CB)$ also remains unchanged. Define $AC(t)$, $CA(t)$, $M_{AC}(t)$, and $W_{AC}(t)$ parametrically to be the values of $AC$, $CA$, $M_{AC}$, and $W_{AC}$ induced by $X_t(v)$.

    As $X(v)$ decreases by $t$, $\|X\|_{\infty}$ and $\|X+Y\|_{\infty}$ increase by no more than $t$, so overall,
    \[
        \|X_t\|_{\infty} + X_t(v) \leqslant \|X\|_{\infty} + t + X(v) - t = \|X\|_{\infty} + X(v),
    \]
    and similarly $\|X_t + Y\|_{\infty} + X_t(v) - Y_t(v) \leqslant \|X+Y\|_{\infty} + X(v) - Y(v)$. Therefore, $Z_{\min_{} }(X_t, Y_t) \leqslant Z_{\min_{} }(X,Y)$. Consequently, $\Phi_R(X_t,Y_t) = \mathbb{E}X_t + (R+1)\cdot \mathbb{E}Y_t +R \cdot\mathbb{E}Z_{\min_{} }(X_t, Y_t) \leqslant \Phi_R(X,Y)$ for all $t$.

    It remains now to find a $t^*$ that attains the equality $f(AC(t^*)) = 1-\lambda$. Observe that at $t=0$, we have $f(AC(0)) \geqslant 1-\lambda$ by assumption; on the other hand, if $t > \|X\|_{\infty}$ then $X_t<0$ everywhere, making every voter prefer $C$ over $A$, at which point $f(AC(t)) = 0$. We argue that if ties are distributed appropriately as $t$ increases, then $f(AC(t))$ is a continuous function. Since $f(AC(0)) \geqslant 1-\lambda$ and $f(AC(t)) = 0$ for $t > \|X\|_{\infty}$, this would imply that there is some value of $t^*$ where $f(AC(t^*)) = 1 - \lambda$.

    Note by \Cref{eq:fw} that a discontinuity of $f(AC(t))$ is either caused by a discontinuity in $\lvert AC(t)\rvert$ or $W_{AC}(t)$. A discontinuity in $\lvert AC(t)\rvert$ happens precisely when $S_t = \{v \mid X_t(v) = 0\}$ has nonzero mass. A discontinuity in $W_{AC}(t)$ can occur when either $S_t$ has nonzero mass, or $P_t = \{(u, v) \in M_{AC}(t) \mid X_t(u) + X_t(v) = 0\}$ has nonzero mass. Fix an arbitrary $t$. Since $X$ is a piecewise step function, let $t' > t$ be the earliest time step after $t$ at which either $S_{t'}$ or $P_{t'}$ has nonzero mass.

    We first handle the case where $P_{t'}$ has nonzero mass. Since $t' > t$, the deliberation for every pair in $P_{t'}$ initially prefers $A$. Now arbitrarily select a subset of pairs $P' \subseteq P_{t'}$ with mass $\varepsilon_1 > 0$. We set the deliberation for every pair in $P'$ to prefer $C$. Clearly $W_{AC}(t')$ decreases by at most $\varepsilon_1$, implying that $W_{AC}(t')$, and hence $f(AC(t'))$, changes continuously.

    \input{figs/fAC_cont}

    We now handle the case where $S_{t'}$ has nonzero mass. Since $t' > t$, every voter in $S_{t'}$ is initially in $AC$. We arbitrarily choose a subset $S' \subseteq S_{t'}$ with mass $\varepsilon_2 > 0$ and set the ties so that every voter in $S'$ is in $CA$. Clearly $\lvert AC\rvert$ decreases by exactly $\varepsilon_2$. To bound the change in $W_{AC}(t')$, first note that because $S_{t'}$ is initially the unique set on which $X_t =0$, no $(A\succ C)$ pairs can take place on it. Consequently, our operation cannot decrease $W_{AC}(t')$. On the other hand, if $W_{AC}(t')$ increased by more than $\varepsilon_2$, then we could remove the set $S'$ and achieve a matching for the original instance with a higher value of $W_{AC}(t')$, which contradicts the original optimality of $W_{AC}(t')$.

    Thus $\lvert AC(t')\rvert$ and $W_{AC}(t')$ are both continuous in the change, so by \Cref{eq:fw}, $f(AC(t'))$ changes continuously for this step as well. This implies there is some value of $t^*$ where $f(AC(t^*)) = 1 - \lambda$.

    The argument for tightening $f(CB)$ is highly analogous. For $t\geqslant 0$, we keep $X_t = X$ unchanged for all $t$ and let $Y_t(v) = Y(v) - t$. Then $\|Y\|_{\infty}$ and $\|X+Y\|_{\infty}$ increase by no more than $t$ as $Y(v)$ decreases by $t$, so
    \[
        \|Y_t\|_{\infty} - Y_t(v) \leqslant \|Y\|_\infty + t - (Y(v)-t) = \|Y\|_{\infty} - Y(v) + 2t,
    \]
    and similarly $\|X_t+Y_t\|_{\infty} + X_t(v) - Y_t(v) \leqslant \|X+Y\|_{\infty} + X(v) - Y(v) + 2t$. By \Cref{eq:Zmin} these imply $Z_{\min_{} }(X_t,Y_t) \leqslant Z_{\min_{} }(X,Y) + t$. Then $\Phi_R(X_t, Y_t) \leqslant \Phi_R(X,Y)$, as the coefficient of $\mathbb{E}Y$ is $R+1$, greater than that of $\mathbb{E}Z$. The rest of the proof mirrors the argument above by noting that when $t > \|Y\|_{\infty}$, we have $f(CB(t)) = 0$ because $Y_t<0$ everywhere. We omit the tie-handling details. \qedhere
\end{proof}

With the two properties of optimal matchings in place, we now restate \programCref{eq:program} below.

\begin{equation}
    \begin{array}{rl}
        \text{Minimize} \qquad& \Phi_R(X,Y)=\mathbb{E}X+(R+1)\cdot\mathbb{E}Y+R\cdot\mathbb{E}Z\\
        \text{over} \qquad& X,Y \text{ on } V,\quad Z=Z_{\min}(X,Y)\ \text{ from \Cref{eq:Zmin}}\\[1em]
        \text{Subject to} \qquad& f(AC)\ \text{is induced by an $A$-optimal $(AC,CA)$ matching}; \\
        & f(CB)\ \text{is induced by a $C$-optimal $(CB,BC)$ matching};\\
        & f(AC)= 1-\lambda,\ \ f(CB)= \lambda.
    \end{array}
    \label{eq:program-reduced}
\end{equation}

We will next show that this program can be made to have a constant number of variables, and it is a bilinear program with two disjoint sets of variables and linear constraints over these.

%% file: figs/matching_Awin_only.tex
\begin{tikzpicture}[x=0.6cm,y=0.6cm,>=stealth,line cap=round,font=\small, yscale=0.75]
    \fill[cyan!10] (0,2) rectangle (2,7);
    \fill[cyan!10] (5,0) rectangle (7,5);
    
    \draw[thick] (0,0) rectangle (2,7);
    \draw[thick] (5,0) rectangle (7,10);
    
    \draw[thick] (0,2)--(2,2);
    \draw[thick] (5,5)--(7,5);
    
    \node at (1,7.7) {$AC$};
    \node at (6,10.7) {$CA$};
    
    \foreach \yL/\yR in {2.5/0.5,3.5/1.5,4.5/2.5,5.5/3.5,6.5/4.5}{
        \draw[thick] (2,\yL) -- (5,\yR);
    }
    \foreach \yL/\yR in {0.5/5.5,1.5/6.5}{
        \draw[dashed] (2,\yL) -- (5,\yR);
    }
    
    \draw[decorate,decoration={brace,amplitude=6pt,raise=0pt}]
    (0,0)--(0,2) node[midway,xshift=-8pt,anchor=east] {irrelevant};
    \draw[decorate,decoration={brace,amplitude=6pt,raise=0pt}]
    (0,2)--(0,7) node[midway,anchor=east] {\begin{tabular}{c} prefix of $AC$ \\ used to enforce $A\succ C$ pairs \end{tabular}};
    \node[anchor=east] at (-0.1,7.2) {\textbf{most pro-$A$ (largest $X(v)$)}};
    
    \draw[decorate,decoration={brace,amplitude=6pt,raise=0pt}]
    (7,10)--(7,5) node[midway,xshift=8pt,anchor=west] {irrelevant};
    \draw[decorate,decoration={brace,amplitude=6pt,raise=0pt}]
    (7,5)--(7,0) node[midway,anchor=west] {\begin{tabular}{c} suffix of $CA$ \\ used to enforce $A\succ C$ pairs \end{tabular}};
    
    \node[anchor=west] at (8.2,10.2) {\textbf{most pro-$C$ (most negative $X(v)$)}};
    \node[anchor=west] at (8.2,0.5) {\textbf{most indifferent ($X(v)$ close to zero)}};
\end{tikzpicture}

%% file: figs/fAC_cont.tex
\begin{figure}[ht]
	\centering
	\definecolor{green}{HTML}{5FA526}
	\definecolor{pink}{HTML}{F08CE2}
	\definecolor{blue}{HTML}{4381FA}
	\definecolor{purple}{HTML}{9161F9}
	
	\begin{subfigure}[t]{0.27\textwidth}
		\centering
		\begin{tikzpicture}
			\begin{axis}[
				width=\linewidth, height=3.2cm,
				xmin=-0.1, xmax=3.1,
				ymin=-1.2, ymax=1.2,
				axis x line=middle,
				axis y line=none,
				ticks=none,
				axis line style={draw=black},
				clip=false,
                xscale = 1.5, yscale = 0.75
				]
				\addplot[green, line width = 2pt] coordinates {(0,1) (1,1)};
				\addplot[blue,  line width = 2pt] coordinates {(1,0) (2,0)};
				\addplot[green, line width = 2pt] coordinates {(2,-1) (3,-1)};
			\end{axis}
		\end{tikzpicture}
		\caption{Original: green segments ($P_{t'}$) are $A\succ C$ pairs; blue segment ($S_{t'}$) is in $AC$.}
	\end{subfigure}%
	\hfill
	\begin{subfigure}[t]{0.27\textwidth}
		\centering
		\begin{tikzpicture}
			\begin{axis}[
				width=\linewidth, height=3.2cm,
				xmin=-0.15, xmax=3.15,
				ymin=-1.2, ymax=1.2,
				axis x line=middle,
				axis y line=none,
				ticks=none,
				axis line style={draw=black},
				clip=false,
                xscale = 1.5, yscale = 0.75
				]
				\addplot[green, line width = 2pt] coordinates {(0,1) (0.67,1)};
				\addplot[red,   line width = 3pt] coordinates {(0.67,1) (1,1)};
				\addplot[blue,  line width = 2pt] coordinates {(1,0) (2,0)};
				\addplot[green,   line width = 2pt] coordinates {(2,-1) (2.67,-1)};
				\addplot[red, line width = 3pt] coordinates {(2.67,-1) (3,-1)};
				
				\node[inner sep=0, outer sep=0, scale=0.9, anchor=south]
				at (axis cs:0.77,1.18) {\scalebox{0.9}{$\phantom{\leftarrow} \varepsilon_1$ }};  
				\node[inner sep=0, outer sep=0, scale=0.9, anchor=north]
				at (axis cs:2.77,-1.2) {\scalebox{0.9}{$\phantom{\leftarrow} \varepsilon_1$ }}; 
			\end{axis}
		\end{tikzpicture}
		\caption{Handling $P_{t'}$. Carve out $P'$ (red) of mass $\varepsilon_1$ (greens) and attribute them to $C\succ A$ pairs.}
	\end{subfigure}%
	\hfill
	\begin{subfigure}[t]{0.28\textwidth}
		\centering
		\begin{tikzpicture}
			\begin{axis}[
				width=\linewidth, height=3.2cm,
				xmin=-0.15, xmax=3.15,
				ymin=-1.2, ymax=1.2,
				axis x line=middle,
				axis y line=none,
				ticks=none,
				axis line style={draw=black},
				clip=false,
                xscale = 1.5, yscale = 0.75
				]
				\addplot[red,    line width = 2pt] coordinates {(0,1) (1,1)};
				\addplot[blue,   line width = 2pt] coordinates {(1,0) (1.67,0)};
				\addplot[purple, line width = 3pt] coordinates {(1.67,0) (2,0)};
				\addplot[red,    line width = 2pt] coordinates {(2,-1) (3,-1)};
				
				\node[inner sep=0, outer sep=0, scale=0.9, anchor=south]
				at (axis cs:1.7,0.22) {$\phantom{\leftarrow} \varepsilon_2$}; 
			\end{axis}
		\end{tikzpicture}
		\caption{Handling $S_{t'}$. Carve out $S'$ of mass $\varepsilon_2$ (purple) and set it to $CA$. Note $W_{AC}(t')$ may change.}
	\end{subfigure}%
    \caption{A visualization of continuous tie-handling on $X$. Left (a): At time step $t'$, both $P_{t'}$ and $S_{t'}$ have positive mass. Middle (b): We handle $P_{t'}$ by continuously allocating increasing subsets $P'\subseteq P_{t'}$ to $C\succ A$ pairs. Right (c): We then handle $S_{t'}$ by continuously allocating increasing subsets $S'\subseteq S_{t'}$ to $CA$ and argue that the change in $W_{AC}(t')$ is also continuous.}
    \label{fig:fAC_cont}
\end{figure}
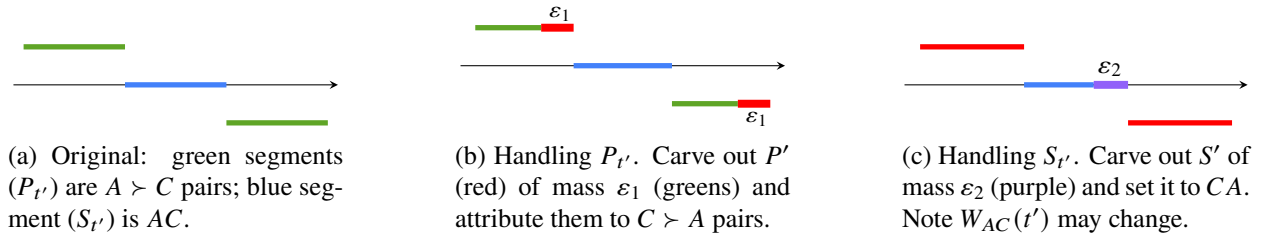

%% file: sections/5_5_LP.tex
\subsection{Bilinear Program and the Distortion of $3$}

\label{subsect:LP}

For the remainder of the section, we fix $\lambda^* = (3 - \sqrt{3})/2 \approx 0.634$ and $w^* = \sqrt{3}-1 \approx 0.732$, and show that the deliberation-via-matching protocol under these parameters has distortion at most $3$. We use this setting, since the construction in \Cref{sec:main_lb} shows that all choices of $(\lambda, w)$ have worst case distortion $\geqslant 3$ for our protocol, thereby proving the optimality of $(\lambda^*, w^*)$. Further, assuming these parameters reduces the number and complexity of cases we need to consider below.

We first show that this setting of parameters, combined with the tightness of the $f$ constraints implies the sizes of $\lvert AC\rvert$ and $\lvert CB\rvert$ satisfy some nice properties. This leads to a collection of instances that we then simplify using the convexity of the variable $Z$ and the max norms into bilinear programs with a constant number of variables, which we can easily solve via vertex enumeration.

\subsubsection[Bounding the Sizes of $\lvert AC\rvert$ and $\lvert CB\rvert$]{Bounding the Sizes of $\lvert AC\rvert$ and $\lvert CB\rvert$}

We begin by characterizing the range of possible sizes of $\lvert AC\rvert$ and $\lvert CB\rvert$ when $f(AC) = 1 - \lambda^*$ and $f(CB) = \lambda^*$. We defer the algebraic derivations to \Cref{app:algebra} to streamline the flow.

\begin{lemma}[Proved in \Cref{app:algebra}]
    \label{lem:mn_mx}
    When $f(AC) = 1 - \lambda^*$, we have $0.25 \leqslant \lvert AC\rvert \leqslant 0.50$. Similarly, when $f(CB) = \lambda^*$, we have $0.50 \leqslant \lvert CB\rvert \leqslant 0.75$. In particular, for instances where $f(AC) = 1 - \lambda^*$ and $f(CB) = \lambda^*$, we always have $\lvert AC\rvert \leqslant \lvert CA\rvert$ and $\lvert BC\rvert \leqslant \lvert CB\rvert$.
\end{lemma}

\begin{lemma}[Proved in \Cref{app:algebra}]
    \label{lem:midpoint}
    When $f(AC) = 1 - \lambda^*$, we have $\lvert AC\rvert + W_{AC} = 0.5$. Similarly, when $f(CB) = \lambda^*$, we have $\lvert BC\rvert + W_{BC} = 0.5$.
\end{lemma}

We now consider the partitions induced on the number line as given in \Cref{tab:interval}. Our goal will be to write a bilinear programming relaxation of \programCref{eq:program-reduced}, where we have variables for each interval in the partition. There are two cases based on the sizes of $\lvert AC\rvert$ and $\lvert BC\rvert$.

\subsubsection[Case 1: $\lvert AC\rvert \leqslant \lvert BC\rvert$]{Bilinear Program Case $1$: $\lvert AC\rvert \leqslant \lvert BC\rvert$}

We first consider the case where $\lvert AC\rvert \leqslant \lvert BC\rvert$. We know from \Cref{lem:mn_mx} that $\lvert AC\rvert \leqslant \lvert CA\rvert$ and $\lvert BC\rvert \leqslant \lvert CB\rvert$, so the partitions induced by $X$ and $Y$ are shown in \Cref{tab:interval}. As shown in \Cref{fig:ac_small_LP}, we partition the range into $9$ intervals labeled $1$ through $9$. The top row of the figure depicts how the $9$ intervals relate to the partition induced by $X$, while the bottom row depicts how the same $9$ intervals relate to the partition induced by $Y$. For the $X$ partition, intervals $1$ and $2$ correspond to $[0, W_{AC}]$, interval $3$ corresponds to $[W_{AC}, \lvert AC\rvert]$, intervals $4$ and $5$ correspond to $[\lvert AC\rvert, \lvert AC\rvert + W_{AC}]$, and intervals $6$ through $9$ correspond to $[\lvert AC\rvert + W_{AC}, 1]$. For the $Y$ partition, interval $1$ corresponds to $[0, W_{BC}]$, intervals $2$ through $4$ correspond to $[W_{BC}, \lvert BC\rvert]$, intervals $5$ and $6$ correspond to $[\lvert BC\rvert,1 - \lvert BC\rvert + W_{BC}]$, and intervals $7$ through $9$ correspond to $[1 - \lvert BC\rvert + W_{BC}, 1]$. The interpretations of each segment of the $X$ and $Y$ partitions are given in \Cref{tab:interval}.

\begin{figure}[htbp]
    \centering
    \scalebox{0.85}{%
  \def\svgwidth{\columnwidth}
  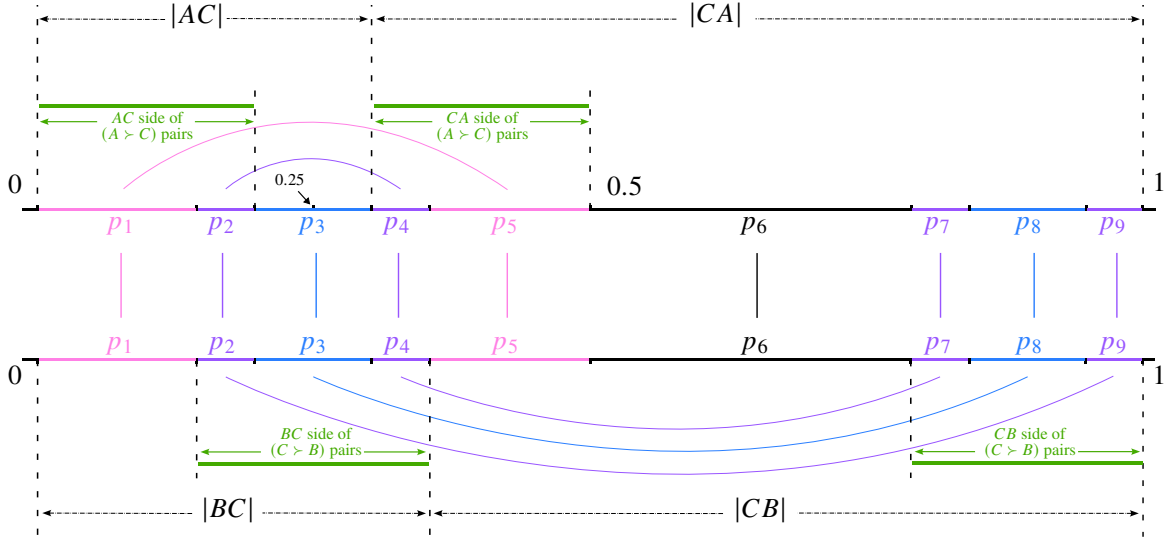
}
    \caption{The lines go from $0$ to $1$, capturing cumulative voter mass. The top line represents $X$ values in decreasing order and the bottom line represents $Y$ values in increasing order. A voter appears at the same position in both lines. The pink masses $p_1$ and $p_5$ represent a set of $A \succ C$ matching pairs. This means $p_1 = p_5$. The same holds for the purple masses $p_2$ and $p_4$. The masses $p_1, p_2, p_3$ correspond to non-negative $X$ values, hence together capture $\lvert AC\rvert$. The masses $p_1, \ldots, p_4$ have non-positive $Y$ values and together capture $\lvert BC\rvert$. The $C \succ B$ pairs are captured by the pairs of masses $(p_2, p_9)$, $(p_3, p_8)$ and $(p_4,p_7)$.}
    \label{fig:ac_small_LP}
\end{figure}

We now show that it suffices to consider instances where for each interval, the values of $X$ and $Y$ are constant across it. In particular, we can replace the voters in an interval with a weighted voter whose $X$ (resp. $Y$) value is equal to the average of the $X$ (resp. $Y$) values of the voters in that interval, and the weight being the mass of voters in that interval. This follows from the lemma below.

\begin{lemma}
    \label{lem:compact}
    For any two voters $u,v$ with values $(X(u), Y(u))$ and $(X(v), Y(v))$, let $\mu = (X(u) + X(v))/2$ and $\nu = (Y(u) + Y(v))/2$. Then if we replace $u$ and $v$ with two voters with identical values $(\mu,\nu)$, the objective in \Cref{eq:linear-obj} does not increase.
\end{lemma}

\begin{proof}
    First, note that
    \begin{align*}
        \max\left(X(u) + Y(u),X(v) + Y(v)\right) &\geqslant \frac{X(u) + Y(u) + X(v) + Y(v)}{2}\\
        &= \mu + \nu \geqslant \min\left(X(u) + Y(u), X(v) + Y(v)\right),
    \end{align*}
    which means $\| X+Y \|_{\infty}$ cannot increase. A similar argument shows that $\| X \|_{\infty}, \| Y \|_{\infty}$ cannot increase. Next, fixing the values of these norms, $Z_{\min}$ in \Cref{eq:Zmin} is the maximum of three linear functions, and is therefore a convex function of $X,Y$. By Jensen's inequality, this means $\mathbb{E}Z$ cannot increase when we replace values by their means. Finally, $\mathbb{E}X, \mathbb{E}Y$ are preserved by this transformation.
\end{proof}

Using this lemma in each interval, we can replace the voters in each interval with a weighted voter with $X,Y$ values equal to the means of $X,Y$ values of voters in the interval. This does not increase the objective in \Cref{eq:linear-obj}. We now show that the $f$ constraints are preserved. Consider for example the intervals $(1,5)$ that define a set of $A\succ C$ matchings, each with non-negative sum of $X$ values. Simple averaging over the pairs of matched voters shows that the sum of the mean values of $X$ in the two intervals is non-negative, so that the new weighted voters also define an $A\succ C$ matching. Further, voters who initially preferred $A$ to $C$ map to a weighted voter with the same preference. This shows the $f$ constraints are preserved in this process.

For each interval $i \in [1, 9]$ let $X_i$ and $Y_i$ denote the uniform value of $X$ and $Y$ respectively on that interval. Also let $p_i$ denote the length of interval $i$. We first show that $p_2 = p_4$. By \Cref{lem:midpoint}, we have $W_{AC} + \lvert AC\rvert = 0.5$, so the midpoint of interval $3$ must be at $0.25$. Similarly, we have $W_{BC} + \lvert BC\rvert = 0.5$ so the midpoint of intervals $2$, $3$, and $4$ collectively must also be at $0.25$. Together, this implies $p_2 = p_4$. Since intervals $2$ through $4$ collectively are centered around $0.25$, and interval $5$ ends at $W_{AC} + \lvert AC\rvert = 0.5$, we also have $p_1 = p_5$. Finally, we define intervals $7$ through $9$ to be the intervals that match with intervals $2$ through $4$ in the $(B, C)$ matching. Thus we have $p_2 = p_9$, $p_3 = p_8$, and $p_4 = p_7$. In total we have the constraints $p_1 = p_5$, $p_2 = p_4 = p_7 = p_9$, and $p_3 = p_8$.

We now describe the constraints on $X$ and $Y$ induced by the matching constraints. Recall that intervals $1$ and $2$ correspond to the section of $AC$ where $A$ wins the deliberation, and intervals $4$ and $5$ correspond to the section of $CA$ where $A$ wins the deliberation. The $A$-optimal matching pairs interval $1$ with interval $5$ and interval $2$ with interval $4$. Thus we have the constraints $X_1 + X_5 \geqslant 0$ and $X_2 + X_4 \geqslant 0$. Similarly for $Y$, we have the constraints $Y_2 + Y_9 \geqslant 0$, $Y_3 + Y_8 \geqslant 0$, and $Y_4 + Y_7 \geqslant 0$. We note that our relaxation will not need to enforce the constraints that $X(u) + X(v) \leqslant 0$ for a deliberation where $C$ wins against $A$ (or the corresponding constraint for $Y$).

By the counter-monotonic coupling of $X$ and $Y$, we have $X_i \geqslant X_{i+1}$ and $Y_i \leqslant Y_{i+1}$ for all $i \in [8]$. Finally, since the section $AC$ corresponds to positive $X$ values and the section $CB$ corresponds to positive $Y$ values, we have $X_3 \geqslant 0$ and $Y_5 \geqslant 0$. Combining everything, we obtain the following relaxation of \programCref{eq:program-reduced}:

\begin{align}
    \label{eq:final_lp}
    \min \qquad & \mathbb{E}X+(R+1)\cdot\mathbb{E}Y+R\cdot\mathbb{E}Z \,, &\\
    \textrm{s.t.} \qquad & \mathbb{E}X = \sum_{i=1}^{9} p_i \cdot X_i, \qquad \mathbb{E}Y = \sum_{i=1}^{9} p_i \cdot Y_i \qquad \mbox{and}
    \qquad \mathbb{E}Z = \sum_{i=1}^{9} p_i \cdot Z_i  \notag \\
    & Z_i \geqslant Z_{\min}(X_i,Y_i) \qquad \forall i  \in [9]  \tag{Set of linear constraints} \\
    & X_i \geqslant X_{i+1} \qquad \mbox{and} \qquad Y_i \leqslant Y_{i+1} \qquad \forall \, i \in [8]
    \tag{Counter-monotonicity} \\
    & X_3 \geqslant 0 \qquad \mbox{and} \qquad Y_5 \geqslant 0 \notag \\
    & X_1 + X_5 \geqslant 0 \qquad \mbox{and} \qquad X_2 + X_4 \geqslant 0 \qquad \tag{$A\succ C$ matchings in $X$}\\
    & Y_2 + Y_9 \geqslant 0, \qquad Y_3 + Y_8 \geqslant 0 \qquad \mbox{and} \qquad Y_4 + Y_7 \geqslant 0 \tag{$C\succ B$ matchings in $Y$} \\
    & \sum_{i=1}^{9} p_i = 1 \qquad \mbox{and} \qquad \sum_{i=1}^{5} p_i = 0.5  \tag{$\lvert AC\rvert + W_{AC} = 0.5$} \\
    & p_1 = p_5, \qquad p_2 = p_4 = p_7 = p_9 \qquad \mbox{and} \qquad p_3 = p_8 \tag{Coupling of masses} \\
    & Z_i, p_i \geqslant 0 \,, \ \forall \, i \in [9] \,. \notag
\end{align}

We note that the constraint $Z_i \geqslant Z_{\min}(X_i, Y_i)$ corresponds to a set of linear inequalities by \Cref{eq:Zmin}, which we write out explicitly in \Cref{app:dual}. Since $\|X\|_{\infty} + X(v) \geqslant 0$, we must have $Z_{i} \geqslant 0$. We include this constraint explicitly in the relaxation to aid in our analysis.

Since the above program contains a multiplicative term when computing the expectation of each variable, it is a bilinear program, where the objective multiplies the $p_i$ variables with the $(X_i,Y_i,Z_i)$ variables, and there are separate linear constraints for the $p_i$ and the $(X_i,Y_i,Z_i)$. In order to solve this program efficiently, we separate the constraints into two parts, where the first one has variables for each $p_i$ and the second one has the remaining variables. If we absorb the $\mathbb{E}X, \mathbb{E}Y, \mathbb{E}Z$ constraints into the objective, the two sets of constraints are disjoint. Since for any fixed $(X_i,Y_i,Z_i)$ variables, the bilinear program is linear in the $p_i$ variables, its optimum is achieved at a vertex of the polytope of the $p_i$. This means the overall optimum is also achieved at such a point, and it therefore suffices to enumerate all extreme points of the first set of constraints (the ones capturing $p_i$) and solve the bilinear program at every such extreme point.

Isolating the $p_i$ variables, we have a polytope defined by the following constraints:
\[
    2p_1 + 4p_2 + 2p_3 + p_6 = 1, \qquad 2p_1 + 2p_2 + p_3 = 0.5, \qquad
    p_1, p_2, p_3, p_6 \geqslant 0
\]
Eliminating $p_3,p_6$, the above reduces to the interior of a triangle on $(p_1,p_2)$ with vertices given by the point set $\{(0,0), (0,0.25), (0.25,0)\}$. Therefore, the $3$ extreme points of the polytope are given by \[(p_1, p_2, p_3, p_6) = \{(0, 0, 0.5, 0), (0, 0.25, 0, 0), (0.25, 0, 0, 0.5)\}.\] For each of the $3$ extreme points, we substitute the $p_i$ variables into \programCref{eq:final_lp} and solve the resulting LP. For $R = 2$, the optimal objective value at each such extreme point is exactly $0$, implying that the maximum distortion is at most $3$. We present the verifiable dual certificates in \Cref{app:dual}.

\subsubsection[Case 2: $\lvert AC\rvert > \lvert BC\rvert$]{Bilinear Program Case 2: $\lvert AC\rvert > \lvert BC\rvert$}

This case uses the same ideas as the previous one. We again have $\lvert AC\rvert \leqslant \lvert CA\rvert$ and $\lvert BC\rvert \leqslant \lvert CB\rvert$, so the partitions induced by $X$ and $Y$ are the same as before. We show the $(A, C)$ and $(C, B)$ matchings and the corresponding set of intervals in \Cref{fig:ac_large_LP}. For the $X$ partition, interval $1$ corresponds to $[0, W_{AC}]$, intervals $2$ through $4$ correspond to $[W_{AC}, \lvert AC\rvert]$, interval $5$ corresponds to $[\lvert AC\rvert, \lvert AC\rvert + W_{AC}]$, and intervals $6$ and $7$ correspond to $[\lvert AC\rvert + W_{AC}, 1]$. For the $Y$ partition, intervals $1$ and $2$ correspond to the segment $[0, W_{BC}]$, interval $3$ corresponds to $[W_{BC}, \lvert BC\rvert]$, intervals $4$ through $6$ correspond to $[\lvert BC\rvert,1 - \lvert BC\rvert + W_{BC}]$, and interval $7$ corresponds to $[1 - \lvert BC\rvert + W_{BC}, 1]$. The interpretations of each segment are given in \Cref{tab:interval}.

\begin{figure}[htbp]
    \centering
    \scalebox{0.85}{%
  \def\svgwidth{\columnwidth}
  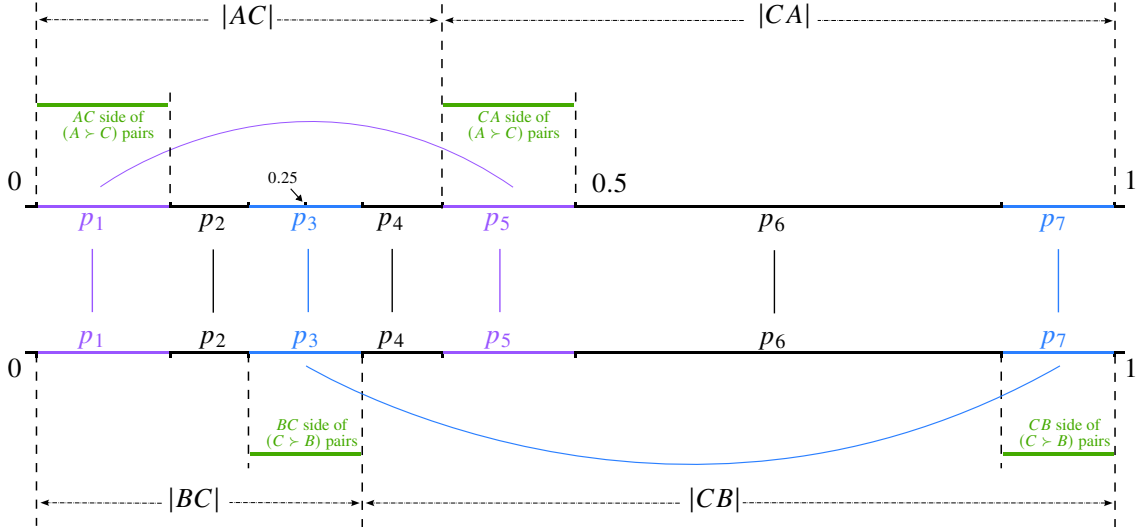
}
    \caption{Interval split for Case 2. The interpretation of this figure is similar to \Cref{fig:ac_small_LP}. Note that analogous to that case, we have $p_2 = p_4$.}
    \label{fig:ac_large_LP}
\end{figure}

Define $X_i$, $Y_i$, and $p_i$ as in the previous case. We first show that $p_2 = p_4$. By \Cref{lem:midpoint}, we have $W_{BC} + \lvert BC\rvert = 0.5$, so the midpoint of interval $3$ must be at $0.25$. Similarly, we have $W_{AC} + \lvert AC\rvert = 0.5$ so the midpoint of intervals $2$, $3$, and $4$ collectively must also be at $0.25$. This implies $p_2 = p_4$. Since intervals $2$ through $4$ collectively are centered around $0.25$, and interval $5$ ends at $\lvert AC\rvert + W_{AC} = 0.5$, we also have $p_1 = p_5$. Finally, we define interval $7$ to be the interval that matches with interval $3$ in the $(B, C)$ matching. Thus $p_3 = p_7$.

For the matching constraints, interval $1$ is mapped with interval $5$ in the $(A, C)$ matching and interval $3$ is mapped with interval $7$ in the $(B, C)$ matching, so we have the constraints $X_1 + X_5 \geqslant 0$ and $Y_3 + Y_7 \geqslant 0$. By the counter-monotonic coupling of $X$ and $Y$, we have $X_i \geqslant X_{i+1}$ and $Y_i \leqslant Y_{i+1}$ for all $i \in [6]$. Finally, since the section $AC$ corresponds to positive $X$ values, and the section $CB$ corresponds to positive $Y$ values, we have $X_4 \geqslant 0$ and $Y_4 \geqslant 0$. Putting everything together, we obtain the following bilinear relaxation:

\begin{align}
    \label{eq:final_lp2}
    \min \qquad & \mathbb{E}X+(R+1)\cdot\mathbb{E}Y+R\cdot\mathbb{E}Z \,, &\\
    \textrm{s.t.} \qquad & \mbox{First three lines of constraints in \programCref{eq:final_lp}} \notag \\
    & X_4 \geqslant 0 \qquad \mbox{and} \qquad Y_4 \geqslant 0 \notag\\
    & X_1 + X_5 \geqslant 0 \qquad \mbox{and} \qquad Y_3 + Y_7 \geqslant 0 \tag{$AC, CB$ matchings} \\
    & \sum_{i=1}^{7} p_i = 1 \qquad \mbox{and} \qquad \sum_{i=1}^{5} p_i = 0.5  \tag{$\lvert AC\rvert + W_{AC} = 0.5$} \\
    & p_1 = p_5, \qquad p_2 = p_4 \qquad \mbox{and} \qquad p_3 = p_7 \tag{Coupling of masses} \\
    & Z_i, p_i \geqslant 0 \,, \ \forall \, i \in [7] \,. \notag
\end{align}

As before, the polytope over $p$ is given by:
\[
    2p_1 + 2p_2 + 2p_3 + p_6 = 1, \qquad 2p_1 + 2p_2 + p_3 = 0.5, \qquad
    p_1, p_2, p_3, p_6 \geqslant 0.
\]

Eliminating $p_3,p_6$, the above again reduces to the interior of a triangle on $(p_1,p_2)$ with vertices $\{(0,0), (0,0.25), (0.25,0)\}$. The polytope therefore has vertices given by \[(p_1, p_2, p_3, p_6) = \{(0, 0, 0.5, 0), (0, 0.25, 0, 0.5), (0.25, 0, 0, 0.5)\}.\] Solving the resulting linear programs again shows that for $R=2$, the objective is at least zero at each extreme point, hence showing the distortion is at most $3$. Again, we present the dual certificates in \Cref{app:dual}. This completes the proof of \Cref{thm:main-informal}.

%% file: figs/ac_small_LP.pdf_tex
\begingroup%
  \makeatletter%
  \providecommand\color[2][]{%
    \errmessage{(Inkscape) Color is used for the text in Inkscape, but the package 'color.sty' is not loaded}%
    \renewcommand\color[2][]{}%
  }%
  \providecommand\transparent[1]{%
    \errmessage{(Inkscape) Transparency is used (non-zero) for the text in Inkscape, but the package 'transparent.sty' is not loaded}%
    \renewcommand\transparent[1]{}%
  }%
  \providecommand\rotatebox[2]{#2}%
  \newcommand*\fsize{\dimexpr\f@size pt\relax}%
  \newcommand*\lineheight[1]{\fontsize{\fsize}{#1\fsize}\selectfont}%
  \ifx\svgwidth\undefined%
    \setlength{\unitlength}{566.92913386bp}%
    \ifx\svgscale\undefined%
      \relax%
    \else%
      \setlength{\unitlength}{\unitlength * \real{\svgscale}}%
    \fi%
  \else%
    \setlength{\unitlength}{\svgwidth}%
  \fi%
  \global\let\svgwidth\undefined%
  \global\let\svgscale\undefined%
  \makeatother%
  \begin{picture}(1,0.46)%
    \lineheight{1}%
    \setlength\tabcolsep{0pt}%
    \put(0,0){\includegraphics[width=\unitlength,page=1]{ac_small_LP.pdf}}%
    \put(0.0876338,0.27092569){\color[rgb]{1,0.50196078,0.89803922}\makebox(0,0)[lt]{\lineheight{1.25}\smash{\begin{tabular}[t]{l}$p_1$\end{tabular}}}}%
    \put(0.17287747,0.27092569){\color[rgb]{0.6,0.33333333,1}\makebox(0,0)[lt]{\lineheight{1.25}\smash{\begin{tabular}[t]{l}$p_2$\end{tabular}}}}%
    \put(0.25018834,0.27092569){\color[rgb]{0.16470588,0.49803922,1}\makebox(0,0)[lt]{\lineheight{1.25}\smash{\begin{tabular}[t]{l}$p_3$\end{tabular}}}}%
    \put(0.32145859,0.27092569){\color[rgb]{0.6,0.33333333,1}\makebox(0,0)[lt]{\lineheight{1.25}\smash{\begin{tabular}[t]{l}$p_4$\end{tabular}}}}%
    \put(0.41300123,0.27092569){\color[rgb]{1,0.50196078,0.89803922}\makebox(0,0)[lt]{\lineheight{1.25}\smash{\begin{tabular}[t]{l}$p_5$\end{tabular}}}}%
    \put(0.62471664,0.27092569){\makebox(0,0)[lt]{\lineheight{1.25}\smash{\begin{tabular}[t]{l}$p_6$\end{tabular}}}}%
    \put(0.77977558,0.27092569){\color[rgb]{0.6,0.33333333,1}\makebox(0,0)[lt]{\lineheight{1.25}\smash{\begin{tabular}[t]{l}$p_7$\end{tabular}}}}%
    \put(0.8570864,0.27092569){\color[rgb]{0.16470588,0.49803922,1}\makebox(0,0)[lt]{\lineheight{1.25}\smash{\begin{tabular}[t]{l}$p_8$\end{tabular}}}}%
    \put(0.92835663,0.27092569){\color[rgb]{0.6,0.33333333,1}\makebox(0,0)[lt]{\lineheight{1.25}\smash{\begin{tabular}[t]{l}$p_9$\end{tabular}}}}%
    \put(0.08720485,0.16891912){\color[rgb]{1,0.50196078,0.89803922}\makebox(0,0)[lt]{\lineheight{1.25}\smash{\begin{tabular}[t]{l}$p_1$\end{tabular}}}}%
    \put(0.17353893,0.16891912){\color[rgb]{0.6,0.33333333,1}\makebox(0,0)[lt]{\lineheight{1.25}\smash{\begin{tabular}[t]{l}$p_2$\end{tabular}}}}%
    \put(0.25084979,0.16891912){\color[rgb]{0.16470588,0.49803922,1}\makebox(0,0)[lt]{\lineheight{1.25}\smash{\begin{tabular}[t]{l}$p_3$\end{tabular}}}}%
    \put(0.32212005,0.16891912){\color[rgb]{0.6,0.33333333,1}\makebox(0,0)[lt]{\lineheight{1.25}\smash{\begin{tabular}[t]{l}$p_4$\end{tabular}}}}%
    \put(0.41366266,0.16891912){\color[rgb]{1,0.50196078,0.89803922}\makebox(0,0)[lt]{\lineheight{1.25}\smash{\begin{tabular}[t]{l}$p_5$\end{tabular}}}}%
    \put(0.62537807,0.16891912){\makebox(0,0)[lt]{\lineheight{1.25}\smash{\begin{tabular}[t]{l}$p_6$\end{tabular}}}}%
    \put(0.78043702,0.16891912){\color[rgb]{0.6,0.33333333,1}\makebox(0,0)[lt]{\lineheight{1.25}\smash{\begin{tabular}[t]{l}$p_7$\end{tabular}}}}%
    \put(0.85774783,0.16891912){\color[rgb]{0.16470588,0.49803922,1}\makebox(0,0)[lt]{\lineheight{1.25}\smash{\begin{tabular}[t]{l}$p_8$\end{tabular}}}}%
    \put(0.92901811,0.16891912){\color[rgb]{0.6,0.33333333,1}\makebox(0,0)[lt]{\lineheight{1.25}\smash{\begin{tabular}[t]{l}$p_9$\end{tabular}}}}%
    \put(0,0){\includegraphics[width=\unitlength,page=2]{ac_small_LP.pdf}}%
    \put(0.00431723,0.30154739){\makebox(0,0)[lt]{\lineheight{1.25}\smash{\begin{tabular}[t]{l}$0$\end{tabular}}}}%
    \put(0.97431221,0.30286922){\makebox(0,0)[lt]{\lineheight{1.25}\smash{\begin{tabular}[t]{l}$1$\end{tabular}}}}%
    \put(0.00431723,0.13899734){\makebox(0,0)[lt]{\lineheight{1.25}\smash{\begin{tabular}[t]{l}$0$\end{tabular}}}}%
    \put(0.61921875,0.02640141){\makebox(0,0)[lt]{\lineheight{1.25}\smash{\begin{tabular}[t]{l}$\lvert CB\rvert$\\\end{tabular}}}}%
    \put(0.97431297,0.13899734){\makebox(0,0)[lt]{\lineheight{1.25}\smash{\begin{tabular}[t]{l}$1$\end{tabular}}}}%
    \put(0.1679331,0.02656521){\makebox(0,0)[lt]{\lineheight{1.25}\smash{\begin{tabular}[t]{l}$\lvert BC\rvert$\\\end{tabular}}}}%
    \put(0,0){\includegraphics[width=\unitlength,page=3]{ac_small_LP.pdf}}%
    \put(0.22611065,0.07878239){\color[rgb]{0.26666667,0.66666667,0}\makebox(0,0)[lt]{\lineheight{1.25}\smash{\begin{tabular}[t]{l}\scalebox{0.6}{\shortstack{$BC$ side of \\ $(C\succ B)$ pairs}}\\\end{tabular}}}}%
    \put(0,0){\includegraphics[width=\unitlength,page=4]{ac_small_LP.pdf}}%
    \put(0.08293661,0.34685578){\color[rgb]{0.26666667,0.66666667,0}\makebox(0,0)[lt]{\lineheight{1.25}\smash{\begin{tabular}[t]{l}\scalebox{0.6}{\shortstack{$AC$ side of \\ $(A\succ C)$ pairs}}\\\end{tabular}}}}%
    \put(0,0){\includegraphics[width=\unitlength,page=5]{ac_small_LP.pdf}}%
    \put(0.14135089,0.44534267){\makebox(0,0)[lt]{\lineheight{1.25}\smash{\begin{tabular}[t]{l}$\lvert AC\rvert$\\\end{tabular}}}}%
    \put(0,0){\includegraphics[width=\unitlength,page=6]{ac_small_LP.pdf}}%
    \put(0.5810281,0.44633941){\makebox(0,0)[lt]{\lineheight{1.25}\smash{\begin{tabular}[t]{l}$\lvert CA\rvert$\\\end{tabular}}}}%
    \put(0,0){\includegraphics[width=\unitlength,page=7]{ac_small_LP.pdf}}%
    \put(0.51197491,0.29936598){\makebox(0,0)[lt]{\lineheight{1.25}\smash{\begin{tabular}[t]{l}$0.5$\end{tabular}}}}%
    \put(0.23054794,0.30775608){\makebox(0,0)[lt]{\lineheight{1.25}\smash{\begin{tabular}[t]{l}\scalebox{0.6}{$0.25$}\\\end{tabular}}}}%
    \put(0,0){\includegraphics[width=\unitlength,page=8]{ac_small_LP.pdf}}%
    \put(0.83137779,0.07904732){\color[rgb]{0.26666667,0.66666667,0}\makebox(0,0)[lt]{\lineheight{1.25}\smash{\begin{tabular}[t]{l}\scalebox{0.6}{\shortstack{$CB$ side of \\ $(C\succ B)$ pairs}}\\\end{tabular}}}}%
    \put(0,0){\includegraphics[width=\unitlength,page=9]{ac_small_LP.pdf}}%
    \put(0.36697933,0.34685578){\color[rgb]{0.26666667,0.66666667,0}\makebox(0,0)[lt]{\lineheight{1.25}\smash{\begin{tabular}[t]{l}\scalebox{0.6}{\shortstack{$CA$ side of \\ $(A\succ C)$ pairs}}\\\end{tabular}}}}%
    \put(0,0){\includegraphics[width=\unitlength,page=10]{ac_small_LP.pdf}}%
  \end{picture}%
\endgroup%

%% file: figs/ac_large_LP.pdf_tex
\begingroup%
  \makeatletter%
  \providecommand\color[2][]{%
    \errmessage{(Inkscape) Color is used for the text in Inkscape, but the package 'color.sty' is not loaded}%
    \renewcommand\color[2][]{}%
  }%
  \providecommand\transparent[1]{%
    \errmessage{(Inkscape) Transparency is used (non-zero) for the text in Inkscape, but the package 'transparent.sty' is not loaded}%
    \renewcommand\transparent[1]{}%
  }%
  \providecommand\rotatebox[2]{#2}%
  \newcommand*\fsize{\dimexpr\f@size pt\relax}%
  \newcommand*\lineheight[1]{\fontsize{\fsize}{#1\fsize}\selectfont}%
  \ifx\svgwidth\undefined%
    \setlength{\unitlength}{581.1023622bp}%
    \ifx\svgscale\undefined%
      \relax%
    \else%
      \setlength{\unitlength}{\unitlength * \real{\svgscale}}%
    \fi%
  \else%
    \setlength{\unitlength}{\svgwidth}%
  \fi%
  \global\let\svgwidth\undefined%
  \global\let\svgscale\undefined%
  \makeatother%
  \begin{picture}(1,0.44878049)%
    \lineheight{1}%
    \setlength\tabcolsep{0pt}%
    \put(0,0){\includegraphics[width=\unitlength,page=1]{ac_large_LP.pdf}}%
    \put(0.06780707,0.25835287){\color[rgb]{0.6,0.33333333,1}\makebox(0,0)[lt]{\lineheight{1.25}\smash{\begin{tabular}[t]{l}$p_1$\end{tabular}}}}%
    \put(0.16958746,0.25791324){\color[rgb]{0,0,0}\makebox(0,0)[lt]{\lineheight{1.25}\smash{\begin{tabular}[t]{l}$p_2$\end{tabular}}}}%
    \put(0.24809004,0.25835287){\color[rgb]{0.16470588,0.49803922,1}\makebox(0,0)[lt]{\lineheight{1.25}\smash{\begin{tabular}[t]{l}$p_3$\end{tabular}}}}%
    \put(0.32069933,0.25835287){\color[rgb]{0,0,0}\makebox(0,0)[lt]{\lineheight{1.25}\smash{\begin{tabular}[t]{l}$p_4$\end{tabular}}}}%
    \put(0.41132808,0.25835287){\color[rgb]{0.6,0.33333333,1}\makebox(0,0)[lt]{\lineheight{1.25}\smash{\begin{tabular}[t]{l}$p_5$\end{tabular}}}}%
    \put(0.6439998,0.25740907){\makebox(0,0)[lt]{\lineheight{1.25}\smash{\begin{tabular}[t]{l}$p_6$\end{tabular}}}}%
    \put(0.88229175,0.25815972){\color[rgb]{0.16470588,0.49803922,1}\makebox(0,0)[lt]{\lineheight{1.25}\smash{\begin{tabular}[t]{l}$p_7$\end{tabular}}}}%
    \put(0.06738858,0.15883424){\color[rgb]{0.6,0.33333333,1}\makebox(0,0)[lt]{\lineheight{1.25}\smash{\begin{tabular}[t]{l}$p_1$\end{tabular}}}}%
    \put(0.17023279,0.15839463){\color[rgb]{0,0,0}\makebox(0,0)[lt]{\lineheight{1.25}\smash{\begin{tabular}[t]{l}$p_2$\end{tabular}}}}%
    \put(0.24873536,0.15883424){\color[rgb]{0.16470588,0.49803922,1}\makebox(0,0)[lt]{\lineheight{1.25}\smash{\begin{tabular}[t]{l}$p_3$\end{tabular}}}}%
    \put(0.32134466,0.15883424){\color[rgb]{0,0,0}\makebox(0,0)[lt]{\lineheight{1.25}\smash{\begin{tabular}[t]{l}$p_4$\end{tabular}}}}%
    \put(0.41197338,0.15883424){\color[rgb]{0.6,0.33333333,1}\makebox(0,0)[lt]{\lineheight{1.25}\smash{\begin{tabular}[t]{l}$p_5$\end{tabular}}}}%
    \put(0.6446451,0.15789044){\makebox(0,0)[lt]{\lineheight{1.25}\smash{\begin{tabular}[t]{l}$p_6$\end{tabular}}}}%
    \put(0.8829371,0.15864109){\color[rgb]{0.16470588,0.49803922,1}\makebox(0,0)[lt]{\lineheight{1.25}\smash{\begin{tabular}[t]{l}$p_7$\end{tabular}}}}%
    \put(0,0){\includegraphics[width=\unitlength,page=2]{ac_large_LP.pdf}}%
    \put(0.00821578,0.28822768){\makebox(0,0)[lt]{\lineheight{1.25}\smash{\begin{tabular}[t]{l}$0$\end{tabular}}}}%
    \put(0.9545524,0.28951727){\makebox(0,0)[lt]{\lineheight{1.25}\smash{\begin{tabular}[t]{l}$1$\end{tabular}}}}%
    \put(0.00821578,0.12964232){\makebox(0,0)[lt]{\lineheight{1.25}\smash{\begin{tabular}[t]{l}$0$\end{tabular}}}}%
    \put(0.58487735,0.01952544){\makebox(0,0)[lt]{\lineheight{1.25}\smash{\begin{tabular}[t]{l}$\lvert CB\rvert$\\\end{tabular}}}}%
    \put(0.95455315,0.12964232){\makebox(0,0)[lt]{\lineheight{1.25}\smash{\begin{tabular}[t]{l}$1$\end{tabular}}}}%
    \put(0.14353002,0.01968533){\makebox(0,0)[lt]{\lineheight{1.25}\smash{\begin{tabular}[t]{l}$\lvert BC\rvert$\\\end{tabular}}}}%
    \put(0,0){\includegraphics[width=\unitlength,page=3]{ac_large_LP.pdf}}%
    \put(0.2268832,0.07089602){\color[rgb]{0.26666667,0.66666667,0}\makebox(0,0)[lt]{\lineheight{1.25}\smash{\begin{tabular}[t]{l}\scalebox{0.6}{\shortstack{$BC$ side of \\ $(C\succ B)$ pairs}}\\\end{tabular}}}}%
    \put(0,0){\includegraphics[width=\unitlength,page=4]{ac_large_LP.pdf}}%
    \put(0.05403387,0.33229833){\color[rgb]{0.26666667,0.66666667,0}\makebox(0,0)[lt]{\lineheight{1.25}\smash{\begin{tabular}[t]{l}\scalebox{0.6}{\shortstack{$AC$ side of \\ $(A\succ C)$ pairs}}\\\end{tabular}}}}%
    \put(0.18860764,0.42943983){\makebox(0,0)[lt]{\lineheight{1.25}\smash{\begin{tabular}[t]{l}$\lvert AC\rvert$\\\end{tabular}}}}%
    \put(0,0){\includegraphics[width=\unitlength,page=5]{ac_large_LP.pdf}}%
    \put(0.64356005,0.4296154){\makebox(0,0)[lt]{\lineheight{1.25}\smash{\begin{tabular}[t]{l}$\lvert CA\rvert$\\\end{tabular}}}}%
    \put(0,0){\includegraphics[width=\unitlength,page=6]{ac_large_LP.pdf}}%
    \put(0.50359099,0.28660911){\makebox(0,0)[lt]{\lineheight{1.25}\smash{\begin{tabular}[t]{l}$0.5$\end{tabular}}}}%
    \put(0.22892865,0.29428492){\makebox(0,0)[lt]{\lineheight{1.25}\smash{\begin{tabular}[t]{l}\scalebox{0.6}{$0.25$}\\\end{tabular}}}}%
    \put(0,0){\includegraphics[width=\unitlength,page=7]{ac_large_LP.pdf}}%
    \put(0.39778833,0.33229833){\color[rgb]{0.26666667,0.66666667,0}\makebox(0,0)[lt]{\lineheight{1.25}\smash{\begin{tabular}[t]{l}\scalebox{0.6}{\shortstack{$CA$ side of \\ $(A\succ C)$ pairs}}\\\end{tabular}}}}%
    \put(0,0){\includegraphics[width=\unitlength,page=8]{ac_large_LP.pdf}}%
    \put(0.86492994,0.0708968){\color[rgb]{0.26666667,0.66666667,0}\makebox(0,0)[lt]{\lineheight{1.25}\smash{\begin{tabular}[t]{l}\scalebox{0.6}{\shortstack{$CB$ side of \\ $(C\succ B)$ pairs}}\\\end{tabular}}}}%
    \put(0,0){\includegraphics[width=\unitlength,page=9]{ac_large_LP.pdf}}%
  \end{picture}%
\endgroup%

%% file: sections/5_6_lower_bounds.tex
\subsection{Complementary Lower Bounds for the Protocol}

\label{sec:main_lb}

To complete the full picture, we now show that our analysis is tight by presenting a few families of parametric instances that yield a lower bound on distortion for any $(\lambda,w)$. This resulting lower bound function attains a minimum value of $3$ at $(\lambda^*,w^*)\approx (0.634, 0.732)$ as defined in \Cref{subsect:LP}, which coincides with the upper bound established in \Cref{thm:main-informal}. Therefore, our specific parameter choices are optimal. 

Three families of instances suffice for this purpose, and we provide a visualization of how their distortions behave as a function of $(\lambda, w)$. We defer the exact algebraic expressions to \Cref{app:algebra}.

\begin{definition}[Permissible Ranges for $\lvert AC\rvert, \lvert CB\rvert$]
        \label{def:min-max}
		For  $\lambda\in (1 /2,1)$ and $w > 0$, with $f(AC) = 1-\lambda$, one must have $AC_{\min} \leqslant \lvert AC\rvert  \leqslant AC_{\max}$, and with $f(CB) = \lambda$, $CB_{\min} \leqslant \lvert CB\rvert  \leqslant CB_{\max}$, where

        		\begin{align*}
						AC_{\min_{} }(\lambda, w) &= \frac{1-\lambda}{1+\lambda w} && CB_{\max_{} }(\lambda, w) = \frac{\lambda(1+w)}{1 + \lambda w}
        				\\[1em]
						CB_{\min_{} }(\lambda, w) &=
        				\begin{cases}
        						\displaystyle \frac{\lambda - (1-\lambda) w}{1 - (1-\lambda)w} & \displaystyle \text{ if }  w \leqslant  \frac{2\lambda-1}{1-\lambda} \\[1em]
        						\displaystyle \displaystyle \frac{\lambda}{1 + (1-\lambda)w} & \displaystyle \text{ if }w > \frac{2\lambda-1}{1-\lambda}
        				\end{cases}
        				&&
						AC_{\max_{} }(\lambda, w) =
        				\begin{cases}
        					\displaystyle \frac{1-\lambda}{1-(1-\lambda)w} & \displaystyle  \text{ if }w \leqslant \frac{2\lambda-1}{1-\lambda} \\[1em]
        					\displaystyle \frac{(1-\lambda)(1+w)}{1+(1-\lambda)w} & \displaystyle \text{ if } w > \frac{2\lambda-1}{1-\lambda}.
        				\end{cases}
        		\end{align*}

\end{definition}

These are the quantities that allow the $f(\boldsymbol\cdot)$ constraints to be satisfied by winning all deliberations (lower bounding the set sizes) or winning zero deliberation (upper bounding the set sizes). Observe that $AC_{\min_{} } + CB_{\max_{} } = AC_{\max_{} } + CB_{\min_{} } = 1$, regardless of $\lambda, w$. The four quantities are found by solving equations. For example, $AC_{\min}$ and $AC_{\max_{} }$ are found by respectively solving:

\begin{minipage}{0.47\textwidth}
		\[
				\begin{aligned}
						\text{find}\quad & \lvert AC\rvert = AC_{\min}  \\
						\text{s.t.}\quad 
						& m_{AC} = \min_{} \{\lvert AC\rvert , \lvert CA\rvert \}, \\
						& 0 \leqslant \lvert AC\rvert  \leqslant 1,\qquad  W_{AC} = m_{AC},\\
						& \lvert AC\rvert  + w \cdot W_{AC} = (1-\lambda) (1 + w\cdot m_{AC})
				\end{aligned}
		\]
\end{minipage}
\hfill
\begin{minipage}{0.47\textwidth}
		\[
				\begin{aligned}
						\text{find}\quad & \lvert AC\rvert = AC_{\max} \\
						\text{s.t.}\quad 
						& m_{AC} = \min_{} \{\lvert AC\rvert , \lvert CA\rvert \}, \\
						& 0 \leqslant \lvert AC\rvert \leqslant 1, \qquad W_{AC} = 0, \\
						& \lvert AC\rvert  + w \cdot W_{AC} = (1-\lambda) (1 + w\cdot m_{AC})
				\end{aligned}
		\]
\end{minipage}
\vspace{10pt} 

The quantities $CB_{\max}, CB_{\min}$ can be computed similarly. We now describe three types of instances. For all three examples, we assume $V$ has unit mass.

\begin{example}[Collinear Points $A-B-C$]
        \label{ex:collinear}
		Embed $V \cup \{A,B,C\}$ on $\mathbb{R}$. Put $A=0, B=1$, and $C=2$. Place voter $v_B$ of mass $AC_{\max_{} }$ at $B$ and $v_C$ with the remaining mass $CB_{\min_{} } = 1 - AC_{\max_{} }$ at $C$. Then:
		\begin{itemize}
				\item $A$ vs. $C$. Arbitrate $v_B$ in favor of $A$. Then $f(AC) = 1-\lambda$ is satisfied by $\lvert AC\rvert  = AC_{\max_{} }$, with $A$ winning zero deliberations.
				\item $C$ vs. $B$. All $(C,B)$ deliberations are ties, and we arbitrate all of them into $C\succ B$ pairs. Then $f(CB) = \lambda$ with $\lvert CB\rvert  = CB_{\min_{} }$ and $C$ winning every deliberation matching.
		\end{itemize}
		This instance has distortion $SC(A) /SC(B) = (AC_{\max_{} } + 2 CB_{\min_{} }) / (CB_{\min_{} })$. This generalizes \Cref{thm:raw-WUS}.
\end{example}

\begin{example}[Co-located $B$ and $C$]
        \label{ex:colocate}
		Embed $V \cup \{A,B,C\}$ on $\mathbb{R}$. Put $A=0$ and $B=C=1$. Place voter $v_A$ of mass $AC_{\min_{} }$ at $A$, and $v_{BC}$ of remaining mass $CB_{\max_{} } = 1 - AC_{\min_{} }$ at $B$ (equivalently $C$). Then:
		\begin{itemize}
				\item $A$ vs. $C$. All $(A,C)$ deliberations are ties and we arbitrate as $A\succ C$. Then $f(AC) = 1-\lambda$ is satisfied by $\lvert AC\rvert  = AC_{\min_{} }$ along with $A$ winning all deliberations.
				\item $C$ vs. $B$. All $(C,B)$ deliberations are also ties; we arbitrate in favor of $B\succ C$. Then $f(CB) = \lambda$ by $\lvert CB\rvert  = CB_{\max_{} }$, with $C$ winning zero deliberation.
		\end{itemize}
		This instance has distortion $SC(A) /SC(B) = CB_{\max_{} } / AC_{\min_{} }$.
\end{example}

\begin{example}[Triangular Instance]
        \label{ex:triangle}
		Embed $A,B,C$ on an equilateral triangle with side length $2$, and partition voters into three point masses of ordinal preferences $ACB, CBA$, and $BAC$, respectively. Define their voter-candidate distances by the following table, 
		where $\eta = 1 - CB_{\min_{} } - AC_{\min_{} } = AC_{\max_{} } - AC_{\min_{} } = CB_{\max_{} } - CB_{\min_{} }$. 
        \begin{table}[htbp!]
				\centering
                \renewcommand{\arraystretch}{1.05}
				\begin{tabular}{|c|c||c|c|c|}
						\hline
						Cluster & Mass & $d(v,A)$ & $d(v,B) $ & $d(v,C) $ \\
						\hline 
						\hline
						$ACB$ & $\eta $ & $1$ & $1$ & $1$ \\
						\hline
						$CBA$ & $CB_{\min_{} }$ & $3$ & $1$ & $1$ \\
						\hline
						$BAC$ & $AC_{\min_{} }$ & $2$ & $0$ & $2$ \\
						\hline
				\end{tabular}
		\end{table}
        
        We note that this instance can be embedded in $(\mathbb{R}^2, \ell_1)$ by placing $A = (0, 0)$, $B=(1, 1)$, $C = (2, 0)$, $ACB = (1, 0)$, $CBA = (2, 1)$, and $BAC = (1, 1)$.
        
		\begin{itemize}
				\item In this instance, $\lvert AC\rvert  = AC_{\max_{} }$ and $\lvert CB\rvert  = CB_{\max_{} }$. 
				\item $A$ vs. $C$. In the $(A,C)$ deliberation, $A$ is unable to win any: either $ACB, BAC$ when paired with $CBA$ results in $C\succ A$. However, because $\lvert AC\rvert  = AC_{\max_{} }$, this is exactly enough to ensure $f(AC) = 1-\lambda$.
				\item $C$ vs. $B$. By the same token, $BAC$ beats both $ACB, CBA$ in the $(C,B)$ deliberation, so every pair outputs $(B\succ C)$. Still, as $\lvert CB\rvert  = CB_{\max_{} }$ we nevertheless reach $f(CB) = \lambda$.
		\end{itemize}
		This instance has distortion
		\[
				\frac{SC(A)}{SC(B)} = \frac{(AC_{\max_{} } - AC_{\min_{} }) + 3 \cdot CB_{\min_{} } + 2 \cdot AC_{\min_{} }}{(AC_{\max_{} } - AC_{\min_{} }) + CB_{\min_{} }}.
		\] 
\end{example}

\noindent\paragraph{The Distortion Lower Bound Over $(\lambda, w)$. } Aggregating \Cref{ex:collinear,ex:colocate,ex:triangle}, we obtain a (piecewise) lower bound of the distortion of our rule with parameters $(\lambda, w)$. 
For each $(\lambda, w)$, we compute the distortions $d_1(\lambda, w)$ from \Cref{ex:collinear}, $d_2(\lambda, w)$ from \Cref{ex:colocate}, and $d_3(\lambda, w)$  from \Cref{ex:triangle}. We then set $\mathcal{D}(\lambda, w) = \max_i d_i(\lambda, w)$ and plot it in \Cref{fig:heatmap}. This creates a $2$D plane of lower bounds of the $(\lambda, w)$ deliberation-via-matching protocol, with global minimizer $(\lambda^*, w^*)$ attaining value $\mathcal{D}(\lambda^*, w^*) = 3$. By \Cref{thm:main-informal}, we conclude that our parameter choice $(\lambda^*, w^*)$ is tight and uniquely optimal. The exact algebraic expressions are provided in \Cref{app:algebra}.

\begin{figure}[htbp!]
    \centering
    \scalebox{0.95}{\input{figs/heatmap.pgf}}
    \caption{Distortion heatmap and the maximizer ($\argmax$) partition of the $(\lambda, w)$-plane induced by $d_1, d_2, d_3$. Each color shows the decision region $\argmax_i d_i(\lambda, w)$; dashed curves represent the decision boundaries $d_i = d_j$. As $d_i$ quickly blows up, we only plot $(\lambda, w)$ over $[0.5,0.7] \times [0,1.25]$. The unique global minimum of $\mathcal{D}(\lambda, w)$ is $3$, attained by $(\lambda^*, w^*)$ (this is also the unique intersection of all three decision boundaries).}
    \label{fig:heatmap}
\end{figure}
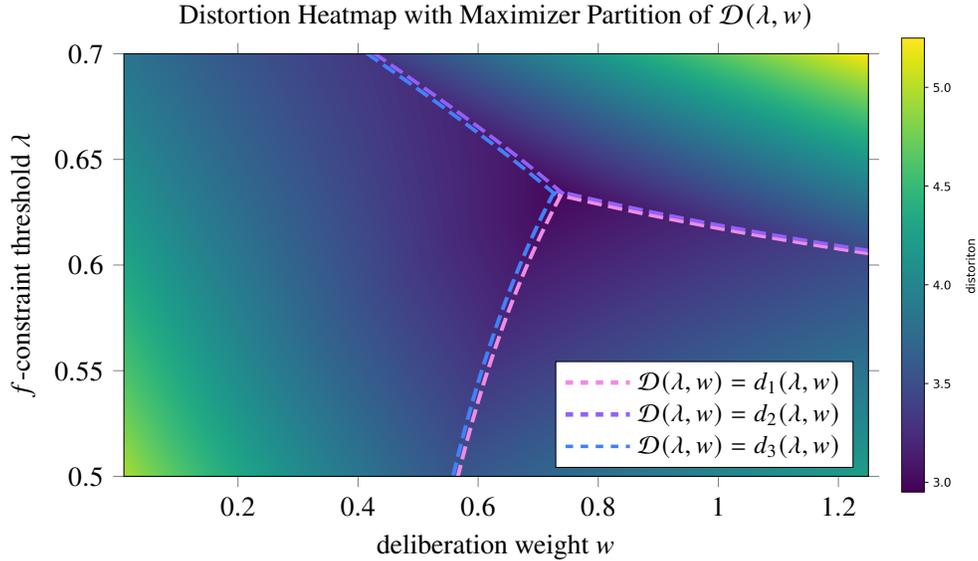

Finally, observe that across all three examples, the lower bounds hold for every maximum cardinality matching, so long as we apportion preference and deliberation ties as described. Therefore, the tightness of \Cref{thm:main-informal} is robust to the \textit{choice} of matchings, which justifies using an \textit{arbitrary} maximum matching in our protocol.

%% file: figs/heatmap.pgf
\begingroup
\providecommand{\Dheatmap}{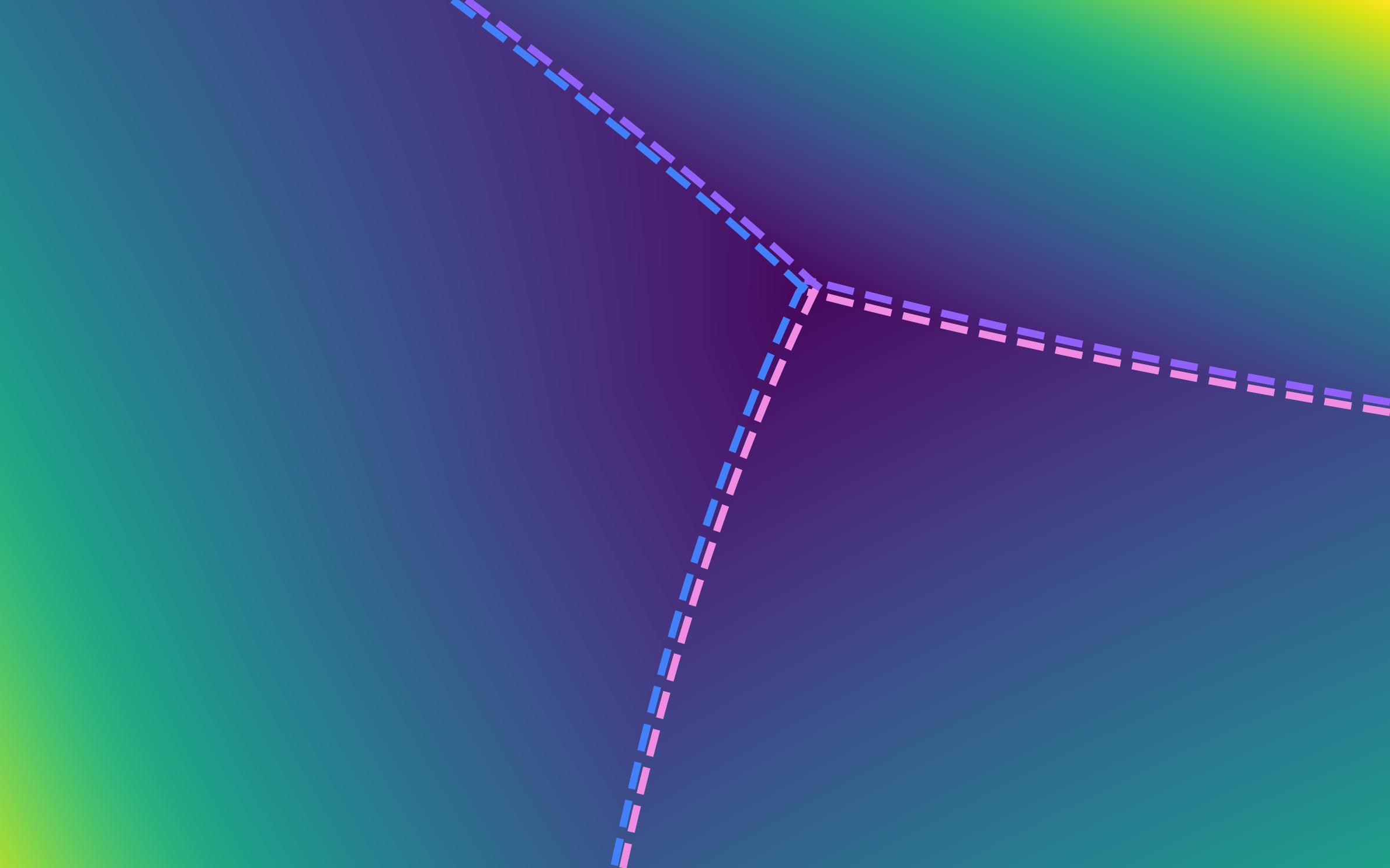}
\providecommand{\Dcolorbar}{figs/heatmap_colorbar.png}

\pgfplotsset{argmaxlegend/.style={
  legend image code/.code={
    \draw[#1] (0cm,0) -- (0.9cm,0);
  },
}}

\begin{tikzpicture}
  \definecolor{dOne}{HTML}{f08ce2}
  \definecolor{dTwo}{HTML}{9161f9}
  \definecolor{dThree}{HTML}{4381fa}

  \begin{axis}[
    width=12.00cm,
    height=7.50cm,
    xmin=0.01, xmax=1.25,
    ymin=0.5, ymax=0.7,
    axis on top,
    enlargelimits=false,
    tick align=outside,
    clip=false, 
    xlabel={deliberation weight $w$},
    ylabel={$f$-constraint threshold $\lambda$},
    title={Distortion Heatmap with Maximizer Partition of $\mathcal{D}(\lambda, w)$},
    legend style={
      at={(rel axis cs:0.98,0.02)},
      anchor=south east,
      draw=black,
      fill=white,
      fill opacity=1,
      font=\small
    },
  ]
    \addplot [forget plot] graphics[xmin=0.01, xmax=1.25, ymin=0.5, ymax=0.7] {\Dheatmap};

    \addlegendimage{argmaxlegend, dashed, draw=dOne, line width=1.8pt}
    \addlegendentry{$\mathcal{D}(\lambda, w)=d_1(\lambda, w)$}
    \addlegendimage{argmaxlegend, dashed, draw=dTwo, line width=1.8pt}
    \addlegendentry{$\mathcal{D}(\lambda, w)=d_2(\lambda, w)$}
    \addlegendimage{argmaxlegend, dashed, draw=dThree, line width=1.4pt}
    \addlegendentry{$\mathcal{D}(\lambda, w)=d_3(\lambda, w)$}

    \node[anchor=west] at (axis description cs:1.03,0.5)
      {\includegraphics[height=6.38cm]{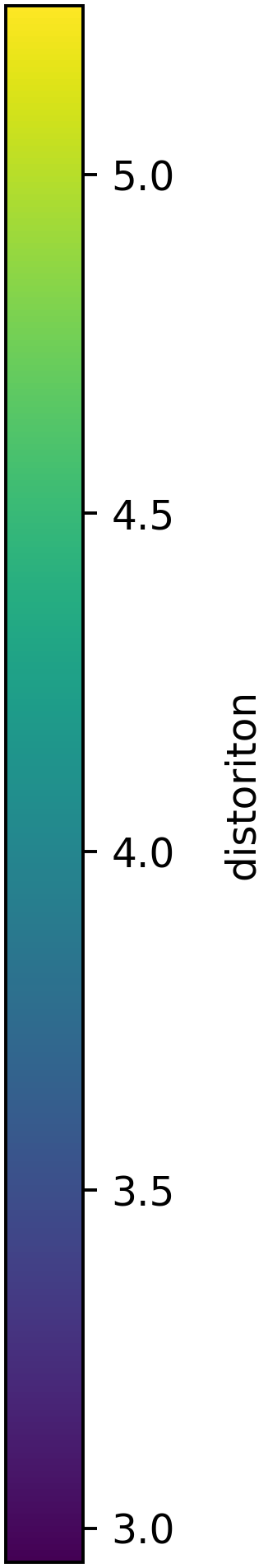}};
  \end{axis}
\end{tikzpicture}
\endgroup

%% file: sections/6_sampling.tex
\section{A Sampling-Based Implementation}
\label{sec:sampling}

In \Cref{sec:general}, we fully analyzed the deliberation-via-matching protocol and proved that it attains distortion $3$ under suitable parameters. From a practical point of view, however, the protocol is admittedly demanding. When $m$ candidates are present, there are $O(m^2)$ distinct pairs of candidates, and the protocol in its current form may require each voter to participate in $O(m^2)$ deliberations.

In this section, we show that the deliberation-via-matching protocol has small sample complexity. We propose a sampling-based implementation that yields a high-probability approximation to the deterministic guarantee (of distortion $3$) with arbitrary accuracy and low per-voter complexity.

In slightly more detail, we demonstrate that for each pair $(A,B)$ of candidates, if we sample a subset of voters whose size grows only logarithmically in the number of candidates, the resulting scores $f(\cdot)$ concentrate around their expectations. We then use these scores in the weighted uncovered set tournament rule in \Cref{sec:protocol} with parameters $(\lambda^*, w^*)$.

\subsection{Estimating $f(\boldsymbol\cdot)$ via the Fractional Matching Score}

The functions $f(\cdot)$ originally depend on the matching (\Cref{eq:f}), but here we need a single population quantity that the sampled sampling is trying to estimate. We begin by replacing the matching-based score with a canonical version obtained by averaging over all maximum matchings. This is what we call the \textbf{fractional matching score}.

Formally, we fix two distinct candidates $A,B$. For any pair of voters $u\in AB$ and $v\in BA$, define
\[
    \mathcal{D}_{u,v} = \mathbf{1} [d(u,A) + d(v,A) \leqslant d(u,B) + d(v,B)],
\]
so $\mathcal{D}_{u,v} = 1$ if the pair $(u,v)$ favors $A$ over $B$ under the deliberation rule. We now define
\begin{equation}
    \label{eq:W-cal}
    \mathcal{W}_{AB} = \frac{1}{\lvert AB\rvert \cdot \lvert BA\rvert} \sum_{u\in AB, v\in BA} \mathcal{D}_{u,v},
\end{equation}
the fraction of disagreeing voter pairs in $AB\times BA$ whose deliberation favors $A$.

\begin{lemma}
    \label{lem:sampling-w}
    Let $M$ be a uniformly random maximum matching between $AB$ and $BA$. Then the expected number of matched pairs in $M$ whose deliberation favors $A$ is equal to $\min \{\lvert AB\rvert, \lvert BA\rvert\} \cdot \mathcal{W}_{AB}$.
\end{lemma}

\begin{proof}
    Let $M$ be a uniformly random maximum matching. For each pair $(u,v)\in AB \times BA$, let $Q_{u,v}$ be the indicator that $(u,v)$ are matched together in $M$. Then the number of $A$-winning matched pairs is $\sum_{u\in AB, v\in BA} \mathcal{D}_{u,v} Q_{u,v}$, so it is enough to estimate $\mathbb{E} [Q_{u,v}]$.

    Fix $u\in AB$ and $v\in BA$. By symmetry, every pair in $AB\times BA$ is matched with the same probability, so it must follow that $\sum_{u\in AB, v\in BA} \mathbb{E} [Q_{u,v}] = \min \{\lvert AB\rvert, \lvert BA\rvert\}.$ There are $\lvert AB\rvert \cdot \lvert BA\rvert$ possible pairs, so
    \[
        \mathbb{E} [Q_{u,v}] = \frac{\min \{\lvert AB\rvert, \lvert BA\rvert\}}{\lvert AB\rvert \cdot \lvert BA\rvert}.
    \]
    Therefore,
    \[
        \mathbb{E} \bigg[ \sum_{u\in AB, v\in BA}  \mathcal{D}_{u,v} Q_{u,v}\bigg] = \sum_{u\in AB, v\in BA} \mathcal{D}_{u,v} \mathbb{E}[Q_{u,v}] = \frac{\min \{\lvert AB\rvert, \lvert BA\rvert\}}{\lvert AB\rvert \cdot \lvert BA\rvert} \sum_{u\in AB, v\in BA} \mathcal{D}_{u,v},
    \]
    from which the claim follows.
\end{proof}

With this convention, we keep everything else analogous to \Cref{eq:f}, except we replace $W_{AB}$ with a quantity related to $\mathcal{W}_{AB}$. With abuse of notation, we now redefine $W_{AB}$ to be the expected mass of $A\succ B$ pairs under a uniformly random maximum matching, so $W_{AB} = \min \{\lvert AB\rvert, \lvert BA\rvert\} \cdot \mathcal{W}_{AB}$. Then, we let
\[
    \mathrm{score}(AB;w) = \frac{\lvert AB \rvert + w \cdot W_{AB}}{n}, \qquad f(AB;w) = \frac{\mathrm{score}(AB;w)}{\mathrm{score}(AB;w) + \mathrm{score}(BA;w)}
\]
so the notation does not change. The only difference is that $W_{AB}$ (resp. $W_{BA}$) no longer represents the number of $A\succ B$ (resp. $B\succ A$) pairs of any particular matching, but the average over all maximum matchings.

This definition gives us a canonical target for sampling. Under this new definition, $f(AB)$ now becomes a fixed quantity determined entirely by the population and the deliberation rule.

Finally, replacing the integral matching count by its fractional version does not change the distortion guarantee proved earlier. The reason is that $W_{AB}$ is additive over matched edges, and any fractional matching can be written as a convex combination of integer matchings. Therefore, if the fractional version satisfies the constraints used in the main analysis, then some integer matching satisfies those same constraints as well. This means \Cref{thm:main-informal} continues to hold if we use the fractional matching score in the protocol in \Cref{sec:protocol}.

\subsection{The Sampling Algorithm}

Fix an ordered pair of candidates $(A,B)$. We now describe how to estimate $f(AB)$. Once this is done for an ordered pair, we apply the same procedure to every candidate pair to obtain an estimate of the entire weighted tournament.

\begin{itemize}
    \item[(1)] Uniformly sample a set $\mathcal{T}$ of $T$ voters from $V$ without replacement.
    \item[(2)] Partition $\mathcal{T}$ into $\mathcal{T}_{AB}$ and $\mathcal{T}_{BA}$ based on their ordinal preferences.
    \item[(3)] Form a uniformly random maximum matching $M_{\mathcal{T}}$ on $\mathcal{T}_{AB}\times \mathcal{T}_{BA}$, for example by uniformly permuting both sets and pairing the first $\min (\lvert \mathcal{T}_{AB}\rvert, \lvert \mathcal{T}_{BA} \rvert)$ pairs.
    \item[(4)] Compute the empirical scores $\hat{f}(AB)$ and $\hat{f}(BA)$ using the outcomes of these deliberations, and use these scores in the weighted uncovered set rule with parameters $(\lambda^*,w^*)$ in \Cref{sec:protocol}. Concretely, let $\widehat{W}_{AB}$ be the number of matched pairs in $M_{\mathcal{T}}$ whose deliberation favors $A$, set $\widehat{\mathrm{score} }(AB;w) = (\lvert \mathcal{T}_{AB}\rvert  + w \cdot \widehat{W}_{AB}) / T$, and compute $\widehat{f}(AB;w) = \widehat{\mathrm{score}}(AB;w) / [\widehat{\mathrm{score} }(AB;w) + \widehat{\mathrm{score} }(BA;w)] $.
\end{itemize}

We show the following theorem. Note that unlike randomized metric distortion bounds \cite{AnshelevichP17,CharikarWRW24}, our bound holds with arbitrarily high probability, and not just in expectation. In other words, the theorem below should be treated as the practical implementation of a deterministic protocol.

\begin{theorem}
    \label{thm:sampling}
    For any $\epsilon, \delta > 0$, let $T = O\left(\frac{1}{\epsilon^2} \log \frac{m}{\delta}\right)$, where $m$ is the number of candidates. Then, with probability at least $1-\delta$, we have $\lvert \hat{f}(AB) - f(AB)\rvert \leqslant \epsilon$ for all pairs of candidates $(A,B)$.
\end{theorem}

The overall sample complexity is $O(m^2 T) = O(m^2 \log  m)$ for constant $\epsilon, \delta$, and is independent of the number of voters $n$. Further, if $n = \omega(m^2 T)$, then any given voter participates in the process (ranking and/or deliberation) for at most one pair $(A,B)$ with high probability. In particular, for constant $\epsilon, \delta$, it suffices that $n = \omega (m^2 \log m)$. This supports our assertion that the protocol has low cognitive complexity.

\begin{proof}[Proof of \Cref{thm:sampling}]
    In the proof below, we will approximate the fractional matching score by sampling. Let $k = \lvert \mathcal{T}_{AB}\rvert$ and $\ell = \lvert \mathcal{T}_{BA}\rvert$ be the realized sizes of the preference groups in the sample. Conditioned on $k$ and $\ell$, the number of deliberative wins in the sample is $\widehat{W} = \sum_{i\in AB, j\in BA}^{} \mathcal{D}_{i,j} Q_{i,j}$, where we recall that $\mathcal{D}_{i,j}$ is the indicator that the pair $(i,j)$ jointly prefers $A$, and $Q_{i,j}$ is the indicator that the pair $(i,j)$ is matched against each other in the sample matching $M_\mathcal{T}$. The empirical score is then given by $\widehat{\mathrm{score} }(AB;w) = (k + w \cdot \widehat{W}) /T$. By symmetry of uniform random matching, for any $i\in AB, j\in BA$, we have
    \[
        \mathbb{E} [Q_{i,j} \mid k,\ell] = \frac{k}{\lvert AB\rvert } \cdot \frac{\ell}{\lvert BA\rvert} \cdot \frac{\min_{} (k,\ell)}{k\cdot \ell} = \frac{\min_{} (k,\ell)}{\lvert AB\rvert \cdot \lvert BA\rvert }
    \]
    where $\lvert AB\rvert, \lvert BA\rvert$ are population sizes. Using \Cref{eq:W-cal}, we see $\mathbb{E}[\widehat{W}\mid k,\ell] = \min_{} (k,\ell) \cdot \mathcal{W}_{AB}$. (Recall $\mathcal{W}_{AB}$ is the fraction of disagreeing voter pairs in $AB\times BA$ whose deliberation favors $A$.)

    Conditioned on group sizes $k$ and $\ell$, and without loss of generality assuming that $k \leqslant \ell$, the quantity $\widehat{W}$ is the sum of the $\mathcal{D}_{i,j}$ values of a random matching of size $k$. This can be obtained by sorting the voters in $AB$ and $BA$ in random order, and pairing the first $k$ voters. We apply Azuma's inequality \cite{Boucheron} to the resulting Doob martingale. At step $t \leqslant k$, consider two executions that are identical till step $t-1$, but the first execution matches voters $(u,v)$ and the second execution matches voters $(u', v')$ at step $t$. We can couple the future executions by considering a random permutation of the unmatched voters in the first execution, and replacing the occurrence of $u'$ (resp. $v'$) with $u$ (resp. $v$) to obtain a coupled run of the second execution. In this coupling, the two executions differ in the $\mathcal{D}_{i,j}$ values of at most three matched edges, so the martingale differences are bounded by $3$. By Azuma's inequality, for any $\epsilon > 0$, we have:
    \[
        \mathbb{P} \left( \bigg \lvert \frac{\widehat{W}}{T} - \frac{\mathbb{E} [\widehat{W}\mid k,\ell]}{T} \bigg \rvert \geqslant  \epsilon \;\Big\vert\; k, \ell \right)  \leqslant  2 \exp  \left( - \frac{(\epsilon T)^2}{18k} \right)  \leqslant 2 \exp  \left( - \frac{\epsilon^2 T}{9} \right)
    \]
    where the final inequality follows from the assumption that $k \leqslant T /2$.

    Next, the group size $k = \lvert \mathcal{T}_{AB}\rvert$ follows a hypergeometric distribution with mean $\mathbb{E}[k] = T /n \cdot \lvert AB\rvert$. By Hoeffding's inequality \cite{hoeffding1963} for sampling without replacement, for any $\epsilon > 0$,
    \[
        \mathbb{P} \left(  \bigg \lvert \frac{k}{T} - \frac{\lvert AB\rvert }{n} \bigg \rvert \geqslant  \epsilon \right) \leqslant  2 \exp  (-2 \epsilon^2 T).
    \]

    Similarly, $\ell /T$ concentrates around $\lvert BA\rvert /n$. Let $\psi  = \min_{} \{\lvert AB\rvert, \lvert BA\rvert \} /n$. Since
    \[
        \bigg \lvert  \frac{\min_{} (k,\ell)}{T} - \psi \bigg \rvert \leqslant  \max_{} \left( \bigg \lvert \frac{k}{T} - \frac{\lvert AB\rvert }{n} \bigg \rvert , \bigg \lvert \frac{\ell}{T} - \frac{\lvert BA\rvert }{n} \bigg \rvert  \right) ,
    \]
    it follows that the expected contribution in the sample, $\mathbb{E} [\widehat{W} /T \mid k, \ell] = \min_{} (k,\ell) / T \cdot \mathcal{W}_{AB}$, concentrates around its population counterpart $\psi \cdot \mathcal{W}_{AB}$ with the same exponential rate:
    \[
        \mathbb{P} \left( \bigg \lvert \frac{\min_{} (k,\ell)}{T} \cdot \mathcal{W}_{AB} - \psi \cdot \mathcal{W}_{AB} \bigg \rvert \geqslant \epsilon \right)  \leqslant  4 \exp  (-2\epsilon^2 T).
    \]

    Combined with the concentration of $\widehat{W}$ around its conditional mean, this ensures that the empirical score $\widehat{\mathrm{score} }(AB)$ is within $O(\epsilon)$ of the population score with high probability. The same holds for $\widehat{\mathrm{score} }(BA)$. Because the denominator of $\widehat{f}(AB)$ (resp. $\widehat{f}(BA)$) is at least $1$, it follows that the empirical $\widehat{f}(AB)$ is within $O(\epsilon)$ of the population $f(AB)$, and the same for $f(BA)$.

    Finally, for a set of $m$ candidates, there are $\binom{m}{2}$ candidate pairs. To ensure that every pairwise score $\widehat{f}$ is within $\epsilon' = O(\epsilon)$ of its expectation, we apply a union bound over all edges in the tournament graph. To achieve this with probability at least $1-\delta$, the required sample size is $T = O(\frac{1}{\epsilon ^{'2}} \log  \frac{m}{\delta}) = O(\frac{1}{\epsilon^2} \log  \frac{m}{\delta})$. This completes the proof of \Cref{thm:sampling}.
\end{proof}

\noindent\paragraph{Approximate Distortion Guarantee.} \Cref{thm:sampling} also yields an approximate distortion bound for a relaxed notion of distortion. Let $\tilde{D}(\lambda_1, \lambda_2,w)$ be the upper bound on metric distortion for a relaxed protocol in \Cref{sec:protocol} where we choose a candidate $A$ such that for all $B$, either:
\begin{itemize}
    \item $f(AB) \geqslant 1 -\lambda_1$, or
    \item there exists $C$ such that $f(AC) \geqslant 1 - \lambda_1$ and $f(CB) \geqslant \lambda_2$.
\end{itemize}

We showed in \Cref{sec:general} that $\tilde{D}(\lambda^*, \lambda^*,w^*) = 3$. Further, from \cite{MunagalaW19}, it is known that such a candidate always exists if $\lambda_1 \geqslant \lambda_2$, since it subsumes the $\lambda_2$-WUS. The proof of \Cref{thm:sampling} shows that with probability $1-\delta$, for the candidate $A$ chosen by $\lambda^*$-WUS for the sample, in the population, for every other candidate $B$, there exists a candidate $C$ such that $f(AC) \geqslant 1 - \lambda^* - \epsilon'$ and $f(CB) \geqslant \lambda^* - \epsilon'$. Therefore, the distortion is at most $\tilde{D}(\lambda^* + \epsilon', \lambda^* - \epsilon', w^*)$. Therefore, the sampling protocol yields a metric distortion bound of $3 + g(\epsilon)$ with probability $1-\delta$, where $g(\cdot)$ is a function that depends on the Lipschitz constant of the function $\tilde{D}(\lambda_1, \lambda_2, w^*)$ around the point $\lambda_1 = \lambda_2 = \lambda^*$.

%% file: sections/7_conclusion.tex
\section{Open Questions}
\label{sec:conclusion}

From the distortion perspective, the immediate open question is to close the gap between the distortion upper bound of $3$ achieved by our protocol and the universal lower bound of $2$ (see \Cref{sec:lb}) for deterministic social choice rules in the same input model, namely, one that takes as input individual preferences together with pairwise deliberation outcomes. Further, can the bounds be improved via a randomized social choice rule? Similarly, our analysis crucially uses the $\lambda$-WUS tournament rule in order to restrict the analysis to three candidates. Can our bilinear relaxation extend to other types of tournament rules, for instance, those considered in \cite{CharikarRT025}?

Stepping back, how can a protocol analogous to matching voters with opposing preferences be extended to deliberating groups of size more than two, and will such an extension also be amenable to bilinear relaxations? Even more broadly, our model assumes opinions do not change during deliberations, and further assumes a specific model of deliberation where the outcome with closer average distance to the voters wins the deliberation. Can our model be extended to handle opinion change in the population via deliberation? Similarly, can our analysis technique be extended to other models of deliberation besides averaging, for instance, the randomized outcome models in \cite{GoelGM25,goyal2025metricdistortionprobabilisticvoting}, or models of deliberation as bargaining with a disagreement (or default) outcome \cite{FainGMS17}?

In sum, we believe that social choice with small-group deliberation is an exciting research direction with the potential for not only deep mathematical analysis, but also the design of novel social choice protocols with practical significance.

%% file: acknowledgement.tex
\paragraph{Acknowledgment:} We used GPT-5 to assist with paraphrasing and clarifying portions of the text, generating code for the optimization routines, and drafting \Cref{app:dual} based on the corresponding dual solution. All content was reviewed, verified, and finalized by the authors, who take full responsibility for the paper’s accuracy and contributions.

%% file: sections/appendix_A_algebra.tex
\section{Omitted Algebraic Derivations from \Cref{sec:general}}
\label[appendix]{app:algebra}

\subsection{Proof of \Cref{lem:mn_mx}}

Recall that since $W_{AC} + W_{CA} = \min(\lvert AC\rvert, 1 - \lvert AC\rvert)$, and $\lvert AC\rvert + \lvert CA\rvert = 1$, we have
\[
    \mbox{score}(AC) = \frac{\lvert AC\rvert + w^* \cdot W_{AC}}{n} \qquad \mbox{and} \qquad \mbox{score}(AC) + \mbox{score}(CA) = \frac{1 + w^* \cdot \min(\lvert AC\rvert, 1 - \lvert AC\rvert)}{n}.
\]
Since $W_{AC} \leqslant \min(\lvert AC\rvert, 1 - \lvert AC\rvert)$, this means
\[
    1-\lambda^* = f(AC) = \frac{\mbox{score}(AC)}{\mbox{score}(AC) + \mbox{score}(CA)} \leqslant \frac{\lvert AC\rvert + w^* \cdot \min(\lvert AC\rvert, 1 - \lvert AC\rvert)}{1 + w^* \cdot \min(\lvert AC\rvert, 1 - \lvert AC\rvert)}.
\]
The RHS is strictly increasing as a function of $\lvert AC\rvert$ when $0 \leqslant \lvert AC\rvert \leqslant 1$. When setting $\lvert AC\rvert = 0.25$, we can verify that the RHS is $1 - \lambda^*$, implying that $0.25$ is the smallest possible value of $\lvert AC\rvert$ to achieve $f(AC) = 1-\lambda^*$.

Similarly, we have $W_{AC} \geqslant 0$, so that
\[
    1-\lambda^* = f(AC) \geqslant \frac{\lvert AC\rvert}{1 + w^* \cdot \min(\lvert AC\rvert, 1 - \lvert AC\rvert)}.
\]
Again, the RHS is strictly increasing as a function of $\lvert AC\rvert$ when $0 \leqslant \lvert AC\rvert \leqslant 1$. Setting $\lvert AC\rvert = 0.5$, we have the RHS is $1 - \lambda^*$, implying that $0.5$ is the largest possible value of $\lvert AC\rvert$.

An identical argument for $\lvert CB\rvert$ shows that $0.5 \leqslant \lvert CB\rvert \leqslant 0.75$.

\subsection{Proof of \Cref{lem:midpoint}}

By \Cref{lem:mn_mx}, we have $\lvert AC\rvert \leqslant \lvert CA\rvert$, so
\[
    1 - \lambda^* = f(AC) = \frac{\lvert AC\rvert + w^* \cdot W_{AC}}{1 + w^* \cdot \lvert AC\rvert}.
\]
Solving for $\lvert AC\rvert$ in terms of $W_{AC}$, we have $\lvert AC\rvert + W_{AC} = 0.5$ as desired. For $f(CB)$, we have by \Cref{lem:mn_mx} that $\lvert CB\rvert \geqslant \lvert BC\rvert$, so
\[
    \lambda^* = f(CB) = \frac{\lvert CB\rvert + w^* \cdot W_{CB}}{1 + w^* \cdot (1 - \lvert CB\rvert)}.
\]
Solving for $W_{CB}$ in terms of $\lvert CB\rvert$, we have $W_{CB} = 1.5 - 2 \lvert CB\rvert$. Substituting $\lvert CB\rvert = 1 - \lvert BC\rvert$ and $W_{CB} = \lvert BC\rvert - W_{BC}$ gives $\lvert BC\rvert + W_{BC} = 0.5$.

\subsection{Closed-form Distortion Lower Bounds in \Cref{sec:main_lb}}

We now derive closed-form lower bounds on the distortion of the $(\lambda, w)$ deliberation-via-matching protocol \\[0.5em] based on \Cref{ex:collinear,ex:colocate,ex:triangle}. Throughout this section, it is convenient to define $\displaystyle \tau(\lambda) = \frac{2\lambda-1}{1-\lambda}$.

\noindent\paragraph{Distortion of \Cref{ex:collinear}.} These instances have distortion $(AC_{\max_{} } + 2 CB_{\min_{} }) / CB_{\min_{} } = 1 + 1 / CB_{\min_{} }$. Hence
\[
    d_1(\lambda, w) =
    \begin{cases}
        1 + \displaystyle \frac{1 - (1-\lambda)w}{\lambda - (1-\lambda)w} = \frac{1 + \lambda - 2 (1-\lambda)w}{\lambda - (1-\lambda) w} & \qquad \text{ if } w \leqslant \tau(\lambda) \\[1.5em]
        1 + \displaystyle \frac{1 + (1-\lambda)w}{\lambda} = \frac{\lambda + 1 + (1-\lambda)w}{\lambda} & \qquad \text{ if } w > \tau(\lambda).
    \end{cases}
\]

\noindent\paragraph{Distortion of \Cref{ex:colocate}.} In these instances, $d_2(\lambda, w) = CB_{\max_{} } / AC_{\min_{} } = \displaystyle \frac{\lambda(1+w)}{1-\lambda}$.

\noindent\paragraph{Distortion of \Cref{ex:triangle}.} We rewrite the fraction by eliminating $CB_{\min}$ using the identity $AC_{\max} + CB_{\min} = 1$:
\begin{align*}
    d_3 &= \frac{\displaystyle 0.5 \cdot (AC_{\max_{} } - AC_{\min_{} }) + 1.5 \cdot CB_{\min_{} } + AC_{\min_{} }}{0.5 \cdot (AC_{\max_{} } - AC_{\min_{} }) + 0.5 \cdot CB_{\min_{} }} \\
    &= \frac{ 3 - 2 AC_{\max_{} } + AC_{\min_{} }}{1 - AC_{\min_{} }}  = \frac{3 - 2AC_{\max_{} } + AC_{\min_{} }}{CB_{\max_{} }},
\end{align*}
Then,
\[
    d_3(\lambda, w) =
    \begin{cases}
        \displaystyle \frac{2 + \lambda + (\lambda^2 + 6\lambda - 4)w - 3\lambda(1-\lambda)w^2}{\lambda(1+w)(1-(1-\lambda)w)} & \qquad \text{ if } w \leqslant \tau(\lambda) \\[1.5em]
        \displaystyle \frac{2 + \lambda + (3 \lambda^2 - 2\lambda + 2)w + (\lambda - \lambda^2)w^2}{\lambda(1+w)(1+(1-\lambda)w)} & \qquad \text{ if } w > \tau(\lambda).
    \end{cases}
\]

Finally, the piecewise analytic lower bound is given by $\mathcal{D}(\lambda, w) = \max_{} \{d_1(\lambda, w), d_2(\lambda, w), d_3(\lambda, w)\}$, and plotting this yields \Cref{fig:heatmap}.

%% file: sections/appendix_B_dual_fitting.tex
\section{Explicit Dual Construction for \Cref{subsect:LP}}
\label[appendix]{app:dual}

In this section we complement the results established in \Cref{subsect:LP} by providing the dual certificates of optimality of the corresponding LPs. We first note that by \Cref{eq:Zmin}, the term $Z_i \geqslant Z_{\min}(X_i, Y_i)$ can be captured by the following set of linear constraints:
\begin{align*}
    & M_X \geqslant X_i,\quad M_X \geqslant -X_i, \quad M_Y \geqslant Y_i,\quad M_Y \geqslant -Y_i & \forall\, i \in [9] \notag , \\
    & M_{X+Y} \geqslant X_i+Y_i,\quad M_{X+Y} \geqslant -(X_i+Y_i) & \forall\, i \in [9] \notag , \\
    & Z_i \geqslant \tfrac12(M_X + X_i), \quad Z_i \geqslant \tfrac12(M_Y - Y_i), \quad Z_i \geqslant \tfrac12\big(M_{X+Y} +X_i-Y_i\big) & \forall\, i \in [9] \notag,
\end{align*}
The variables $M_X$, $M_Y$, and $M_{X+Y}$ represent $\|X\|_{\infty}$, $\|Y\|_{\infty}$, and $\|X + Y\|_{\infty}$ respectively.

\subsection{Dual Certificates for Case 1}

\subsubsection[Vertex (p1,p2,p3,p6)=(0,0,0.5,0)]{Vertex $(p_1,p_2,p_3,p_6)=(0,0,0.5,0)$}

We consider the LP obtained by substituting $p_3=p_8=0.5$ into \programCref{eq:final_lp}. The primal objective (for $R=2$) is
\[
    \Phi_2 \;=\; \tfrac12(X_3+X_8) \;+\; \tfrac32(Y_3+Y_8) \;+\; (Z_3+Z_8).
\]

A valid dual certificate is given by the following nonnegative multipliers on the displayed constraints:
\[
    \begin{array}{rcl}
        0.5 &\text{on}& M_{X+Y} + X_8 + Y_8 \geqslant 0,\\[2pt]
        1   &\text{on}& Z_3 - \tfrac12 M_{X+Y} + \tfrac12 Y_3 - \tfrac12 X_3 \geqslant 0,\\[2pt]
        1   &\text{on}& Z_8 \geqslant 0,\\[2pt]
        1   &\text{on}& Y_3 + Y_8 \geqslant 0,\\[2pt]
        1   &\text{on}& X_3 \geqslant 0.
    \end{array}
\]

Adding the weighted inequalities (left-hand sides minus right-hand sides) with these multipliers gives
\[
    \begin{aligned}
        &\;0.5\big(M_{X+Y}+X_8+Y_8\big)
        +1\big(Z_3-\tfrac12 M_{X+Y}+\tfrac12 Y_3-\tfrac12 X_3\big)\\
        &\qquad\qquad +1\cdot Z_8 + 1\big(Y_3+Y_8\big) + 1\cdot X_3
        \;\geqslant\; 0.
    \end{aligned}
\]
Collecting terms on the left-hand side, all $M$-terms cancel and the sum simplifies exactly to
\[
    \tfrac12 X_3 + \tfrac12 X_8 + \tfrac32 Y_3 + \tfrac32 Y_8 + Z_3 + Z_8 \;=\; \Phi_2.
\]
Hence $\Phi_2 \geqslant 0$ for every primal feasible point.

\subsubsection[Vertex (p1,p2,p3,p6)=(0,0.25,0,0)]{Vertex $(p_1,p_2,p_3,p_6)=(0,0.25,0,0)$}

For $V_2$ the mass pattern places $p=0.25$ on indices $2,4,7,9$. The objective (for $R=2$) becomes
\[
    \Phi_2 \;=\; 0.25(X_2+X_4+X_7+X_9)
    \;+\; 0.75(Y_2+Y_4+Y_7+Y_9)
    \;+\; 0.5(Z_2+Z_4+Z_7+Z_9).
\]

A valid choice of multipliers (all nonnegative) on the primal constraints is:
\[
    \begin{array}{rcl}
        0.25 &\text{on}& M_Y + Y_4 \geqslant 0,\\[2pt]
        0.25 &\text{on}& M_{X+Y} + X_7 + Y_7 \geqslant 0,\\[2pt]
        0.25 &\text{on}& M_{X+Y} + X_9 + Y_9 \geqslant 0,\\[2pt]
        0.5  &\text{on}& Z_2 - \tfrac12 M_{X+Y} + \tfrac12 Y_2 - \tfrac12 X_2 \geqslant 0,\\[2pt]
        0.5  &\text{on}& Z_4 - \tfrac12 M_{X+Y} + \tfrac12 Y_4 - \tfrac12 X_4 \geqslant 0,\\[2pt]
        0.5  &\text{on}& Z_7 - \tfrac12 M_Y + \tfrac12 Y_7 \geqslant 0,\\[2pt]
        0.5  &\text{on}& Z_9 \geqslant 0,\\[2pt]
        0.5  &\text{on}& X_2 + X_4 \geqslant 0,\\[2pt]
        0.5  &\text{on}& Y_2 + Y_9 \geqslant 0,\\[2pt]
        0.25 &\text{on}& Y_4 + Y_7 \geqslant 0.
    \end{array}
\]

Summing these weighted inequalities yields on the left hand side
\[
    0.25 X_2 + 0.25 X_4 + 0.25 X_7 + 0.25 X_9
    + 0.75 Y_2 + 0.75 Y_4 + 0.75 Y_7 + 0.75 Y_9
    + 0.5 Z_2 + 0.5 Z_4 + 0.5 Z_7 + 0.5 Z_9,
\]
which is precisely $\Phi_2$. Thus $\Phi_2\geqslant 0$.

\subsubsection[Vertex (p1,p2,p3,p6)=(0.25,0,0,0.5)]{Vertex $(p_1,p_2,p_3,p_6)=(0.25,0,0,0.5)$}

For $V_3$ we use the mass assignment $p_1=0.25,\ p_5=0.25,\ p_6=0.5$. The objective (for $R=2$) is
\[
    \Phi_2 \;=\; 0.25(X_1+X_5) + 0.5 X_6
    \;+\; 0.75(Y_1+Y_5) + 1.5 \,Y_6
    \;+\; 0.5(Z_1+Z_5) + 1.0\,Z_6.
\]

A valid set of nonnegative multipliers is
\[
    \begin{array}{rcl}
        0.5 &\text{on}& M_Y + Y_1 \geqslant 0,\\[2pt]
        0.5 &\text{on}& M_{X+Y} + X_6 + Y_6 \geqslant 0,\\[2pt]
        0.5 &\text{on}& Z_1 - \tfrac12 M_{X+Y} + \tfrac12 Y_1 - \tfrac12 X_1 \geqslant 0,\\[2pt]
        0.5 &\text{on}& Z_5 - \tfrac12 M_{X+Y} + \tfrac12 Y_5 - \tfrac12 X_5 \geqslant 0,\\[2pt]
        1   &\text{on}& Z_6 - \tfrac12 M_Y + \tfrac12 Y_6 \geqslant 0,\\[2pt]
        0.5 &\text{on}& Y_6 - Y_5 \geqslant 0,\\[2pt]
        0.5 &\text{on}& X_1 + X_5 \geqslant 0,\\[2pt]
        1   &\text{on}& Y_5 \geqslant 0.
    \end{array}
\]

Summing these weighted inequalities yields on the left hand side
\[
    0.25 X_1 + 0.25 X_5 + 0.5 X_6
    + 0.75 Y_1 + 0.75 Y_5 + 1.5 Y_6
    + 0.5 Z_1 + 0.5 Z_5 + 1.0 Z_6,
\]
which equals $\Phi_2$, and therefore $\Phi_2\geqslant0$.

\subsection{Dual certificates for Case 2}

\subsubsection[Vertex (p1,p2,p3,p6)=(0,0,0.5,0)]{Vertex $(p_1,p_2,p_3,p_6)=(0,0,0.5,0)$}

We use the mass assignment $p_3 = p_7 = 0.5$. The objective is
\[
    \Phi_2 = \tfrac12 X_3 + \tfrac12 X_7 + \tfrac32 Y_3 + \tfrac32 Y_7 + Z_3 + Z_7.
\]

A valid dual certificate is given by the following nonnegative multipliers on the displayed constraints:
\[
    \begin{array}{rcl}
        0.5 &\text{on}& M_X + X_7 \geqslant 0,\\[3pt]
        1   &\text{on}& Z_3 - \tfrac12 M_X - \tfrac12 X_3 \geqslant 0,\\[3pt]
        1   &\text{on}& X_3 - X_4 \geqslant 0,\\[3pt]
        1.5 &\text{on}& Y_3 + Y_7 \geqslant 0,\\[3pt]
        1   &\text{on}& X_4 \geqslant 0 ,\\[3pt]
        1   &\text{on}& Z_7 \geqslant 0 .
    \end{array}
\]

Multiply and sum these inequalities with the listed multipliers. On the left-hand side the $M$-terms cancel:
\[
    \begin{aligned}
        &\;0.5(M_X+X_7) \;+\; 1\big(Z_3 - \tfrac12 M_X - \tfrac12 X_3\big)
        \;+\; 1(X_3 - X_4) \\
        &\qquad\quad +\;1.5(Y_3+Y_7) \;+\; 1\cdot X_4 \;+\; 1\cdot Z_7 \;\geqslant\; 0.
    \end{aligned}
\]
Grouping and simplifying the left hand side yields exactly
\[
    \tfrac12 X_3 + \tfrac12 X_7 + \tfrac32 Y_3 + \tfrac32 Y_7 + Z_3 + Z_7
    \;=\; \Phi_2.
\]
Thus $\Phi_2\geqslant 0$.

\subsubsection[Vertex (p1,p2,p3,p6)=(0,0.25,0,0.5)]{Vertex $(p_1,p_2,p_3,p_6)=(0,0.25,0,0.5)$}

We use the mass assignment $p_2=p_4=0.25, p_6=0.5$. A valid dual certificate is given by the following nonnegative multipliers on the displayed constraints:
\[
    \begin{array}{rcl}
        0.5 &\text{on}& M_Y + Y_2 \;\geqslant\; 0,\\[3pt]
        0.5 &\text{on}& M_{X+Y} + X_6 + Y_6 \;\geqslant\; 0,\\[3pt]
        0.5 &\text{on}& Z_2 - \tfrac12 M_{X+Y} + \tfrac12 Y_2 - \tfrac12 X_2 \;\geqslant\; 0,\\[3pt]
        0.5 &\text{on}& Z_4 - \tfrac12 M_{X+Y} + \tfrac12 Y_4 - \tfrac12 X_4 \;\geqslant\; 0,\\[3pt]
        1        &\text{on}& Z_6 - \tfrac12 M_Y + \tfrac12 Y_6 \;\geqslant\; 0,\\[3pt]
        0.5 &\text{on}& X_2 - X_3 \;\geqslant\; 0,\\[3pt]
        0.5 &\text{on}& X_3 - X_4 \;\geqslant\; 0,\\[3pt]
        0.5 &\text{on}& Y_5 - Y_4 \;\geqslant\; 0,\\[3pt]
        0.5 &\text{on}& Y_6 - Y_5 \;\geqslant\; 0,\\[3pt]
        1        &\text{on}& X_4 \;\geqslant\; 0,\\[3pt]
        1        &\text{on}& Y_4 \;\geqslant\; 0.
    \end{array}
\]

Summing these weighted inequalities yields on the left hand side
\[
    0.25(X_2+X_4) + 0.5 X_6
    + 0.75(Y_2+Y_4) + 1.5 Y_6
    + 0.5(Z_2+Z_4) + 1.0 Z_6.
\]
This is exactly
\[
    \Phi_2 = \mathbb{E}X + 3\mathbb{E}Y + 2\mathbb{E}Z
\]
under the mass assignment $p_2=p_4=0.25,\ p_6=0.5$. Therefore $\Phi_2 \geqslant 0$ for this vertex.

\subsubsection[Vertex (p1,p2,p3,p6)=(0.25,0,0,0.5)]{Vertex $(p_1,p_2,p_3,p_6)=(0.25,0,0,0.5)$}

We use the mass assignment $p_1=p_5=0.25,\ p_6=0.5$. A valid dual certificate is given by the following nonnegative multipliers on the displayed constraints:
\[
    \begin{array}{rcl}
        0.5 &\text{on}& M_Y + Y_1 \;\geqslant\; 0,\\[3pt]
        0.5 &\text{on}& M_{X+Y} + X_6 + Y_6 \;\geqslant\; 0,\\[3pt]
        0.5 &\text{on}& Z_1 - \tfrac12 M_{X+Y} + \tfrac12 Y_1 - \tfrac12 X_1 \;\geqslant\; 0,\\[3pt]
        0.5 &\text{on}& Z_5 - \tfrac12 M_{X+Y} + \tfrac12 Y_5 - \tfrac12 X_5 \;\geqslant\; 0,\\[3pt]
        1        &\text{on}& Z_6 - \tfrac12 M_Y + \tfrac12 Y_6 \;\geqslant\; 0,\\[3pt]
        0.5 &\text{on}& Y_6 - Y_5 \;\geqslant\; 0,\\[3pt]
        0.5 &\text{on}& X_1 + X_5 \;\geqslant\; 0,\\[3pt]
        1        &\text{on}& Y_5 \;\geqslant\; 0.
    \end{array}
\]

Summing these weighted inequalities yields on the left hand side
\[
    0.25(X_1+X_5) + 0.5 X_6
    + 0.75(Y_1+Y_5) + 1.5 Y_6
    + 0.5(Z_1+Z_5) + 1.0 Z_6,
\]
which equals the objective
\[
    \Phi_2 = \mathbb{E}X + 3\mathbb{E}Y + 2\mathbb{E}Z
\]
for the mass assignment $p_1=p_5=0.25,\ p_6=0.5$. Thus $\Phi_2\geqslant0$ for this vertex.

Therefore, for each of the six extreme points of the $p$-polytope in Case 1 and 2, the nonnegative dual multipliers above satisfy
\[
    \sum_j \lambda_j(\text{LHS}_j-\text{RHS}_j) = \Phi_2,
\]
hence each yields an analytic dual certificate proving $\Phi_2 \geqslant 0$ and therefore a distortion bound of at most $3$.